\documentclass[reqno,12pt]{amsart}

\usepackage{color} 
\usepackage[dvipsnames]{xcolor}
\usepackage[pdftex]{graphicx}
\usepackage[pdftex]{hyperref}
\usepackage{hyperref}
\hypersetup{
    unicode=false,          
    pdftoolbar=true,        
    pdfmenubar=true,        
    pdffitwindow=false,     
    pdfstartview={FitH},    
    pdftitle={Article},     
    pdfauthor={Collaboration},   
    pdfsubject={Mathematics},    
    pdfcreator={PI},             
    pdfproducer={PI},            
    pdfkeywords={PDE, analysis}, 
    pdfnewwindow=true,      
    colorlinks=true,        
    linkcolor=black,          
    citecolor=blue,         
    filecolor=magenta,      
    urlcolor=cyan           
}

\usepackage[color=green!20]{todonotes}

\usepackage{times}
\usepackage{amsfonts}
\usepackage{amsmath}
\usepackage{amsthm}
\usepackage{amssymb}
\usepackage{bbm}
\usepackage{amsbsy}
\usepackage{amscd}
\usepackage{enumerate}
\usepackage{verbatim}
\usepackage{subfigure}
\usepackage{cleveref}
\usepackage{changepage}
\usepackage[algo2e,linesnumbered,ruled]{algorithm2e}

\newtheorem{theorem}{Theorem}[section]

\newtheorem{remark}[theorem]{Remark}

\newcommand{\abeta}[0]{amyloid-$\beta$}

\newcommand{\beginsupplement}{%
        \setcounter{table}{0}
        \renewcommand{\thetable}{S\arabic{table}}%
        \setcounter{figure}{0}
        \renewcommand{\thefigure}{S\arabic{figure}}%
     }

\numberwithin{equation}{section}  

\usepackage[english]{babel}
\usepackage{chemformula}
{}  
\newtheorem{thm}{Theorem}[section]

\newtheorem{prop}[theorem]{Proposition}

\begin{document}

\vspace*{-0.9cm}

\title
[Stability analysis of a membrane-protein clustering model]
{Stability analysis of a bulk-surface reaction model for membrane-protein clustering}

\author[L.M. Stolerman]{Lucas M. Stolerman$^{1*}$}
\author[M. Getz]{Michael Getz$^{1*}$}
\author[S.G. Llewellyn Smith]{Stefan G. Llewellyn Smith$^{1,2}$}
\author[M. Holst]{Michael Holst$^{3,4}$}
\author[P. Rangamani]{Padmini Rangamani$^{1**}$}

\thanks{$^1$Department of Mechanical and Aerospace Engineering, University of California, San Diego, La Jolla CA 92093-0411}
\thanks{$^2$Scripps Institution of Oceanography, University of California, San Diego, La Jolla, CA 92093-0213}
\thanks{$^3$Department of Mathematics, University of California, San Diego, La Jolla, CA 92093-0112}
\thanks{$^4$Department of Physics, University of California, San Diego, La Jolla, CA 92093-0424}
\thanks{$^*$ Both these authors contributed equally}
\thanks{$^{**}$To whom correspondence should be addressed. e-mail: prangamani@ucsd.edu}
 
\date{\today}


\begin{abstract}
Protein aggregation on the plasma membrane (PM) is of critical importance to many cellular processes such as cell adhesion, endocytosis, fibrillar conformation, and vesicle transport. Lateral diffusion of protein aggregates or clusters on the surface of the PM plays an important role in governing their heterogeneous surface distribution. However, the stability behavior of the surface distribution of protein aggregates remains poorly understood. Therefore, understanding the spatial patterns that can emerge on the PM solely through protein-protein interaction, lateral diffusion, and feedback is an important step towards a complete description of the mechanisms behind protein clustering on the cell surface.  In this work, we investigate the pattern formation of a reaction-diffusion model that describes the dynamics of a system of ligand-receptor complexes. The purely diffusive ligand in the cytosol can bind receptors in the PM, and the resultant ligand-receptor complexes not only diffuse laterally but can also form clusters resulting in different oligomers. Finally, the largest oligomers recruit ligands from the cytosol in a positive feedback. From a methodological viewpoint, we provide theoretical estimates for diffusion-driven instabilities of the protein aggregates based on the Turing mechanism. Our main result is a threshold phenomenon, in which a sufficiently high recruitment of ligands promotes the input of new monomeric components and consequently drives the formation of a single-patch spatially heterogeneous steady-state.

\end{abstract}

\maketitle

\vspace*{-0.8cm}
{\scriptsize
\tableofcontents
}
\vspace*{-1.2cm}

\clearpage

\section{Introduction}
\noindent

Biological membranes are unique two-dimensional structures that separate cellular contents from the extracellular environment and regulate the transport of material into and out of the cell \cite{darnell1990molecular,stillwell2013introduction}. 
In addition to lipids and carbohydrates, these membranes contain a large proportion of proteins, the composition of which depends on the cell type \cite{stillwell2013introduction,guidotti1972composition,yeagle2011structure,albersheim1975carbohydrates,jain1988introduction}. 
One of the interesting features of membrane proteins is their ability to form clusters on the cell surface \cite{hashimoto2010,johannes2018,ispolatov2005}. This clustering of proteins on the plasma membrane (PM) results in a spatial heterogeneity in the distribution of protein densities. Many factors can induce such a spatial heterogeneity, including lateral diffusion, physical barriers from the cytoskeleton \cite{porat-shliom2013}, lipid raft affinity \cite{lorent2017}, and curvature differences along the membrane\cite{johannes2018}.
The formation of protein clusters is intimately related to various cellular phenomena such as polarization, membrane depolarization, receptor signaling, enzyme activity, and  cytoskeletal regulation \cite{Mori2008,lao2010,lemmon2010,sleno2018,baisamy2005,chen2005}.

A particular example of proteins forming clusters on the membrane is well-elucidated by amyloid-$\beta$ aggregation/fibrillation in the context of Alzheimer's disease. 
It is thought that amyloid-$\beta$ can become cytotoxic is when it aggregates on the membrane at high levels \cite{askarova2011}. 
Biophysical measurements show that amyloid-$\beta$ aggregates become more stable when oligomerized on the membrane surface \cite{sarkar2013,zhang2012} and also can destabilize certain membrane compositions \cite{andreasen2015}. It is also thought that membrane components such as cholesterol may initiate aggregation of \abeta, which may then be bolstered by as a yet-unidentified secondary feedback mechanism \cite{Habchi2018}.In addition to \abeta, surface receptors such as $\alpha$-amino-3-hydroxy-5-methyl-4-isoxazolepropionic acid receptor (AMPAR) \cite{choquet2010,gan2015} and membrane-bound kinases such as Fyn \cite{padmanabhan2019} are known to cluster on the membrane; these clusters have been implicated in neuronal functioning in physiology \cite{lao2010,lorent2017} and disease \cite{andreasen2015,askarova2011}.

One of the open questions in the field of protein aggregations is the role of the spatial organization of membrane proteins due to bulk-surface reactions and feedback mechanisms.  Mathematical modeling has provided substantial insight into the geometric coupling of bulk-surface reaction-diffusion systems \cite{Rangamani2013,Frey2018,Denk2018}, including wave-pinning formulations  \cite{Mori2008,Cusseddu2018}, spatial patterning  \cite{Giese2015,Diegmiller2018}, and generalized stability analysis \cite{Madzvamuse2015,Ratz2012, Ratz2015}. From a modeling perspective, several authors have proposed the classical Smoluchowski coagulation model \cite{Smoluchowski1918,Drake1972} as a suitable candidate for describing protein aggregation. The use of different aggregation kernels \cite{Arosio2012, Zidar2018} allowed a successful combination of experimental measurements with computational predictions. These models performed well in terms of comparisons to data and estimation of kinetic parameters such as the aggregation time and the asymptotic cluster distribution. However, by using the original Smoluchowski  systems of ordinary differential equations (ODEs), these studies lack descriptions of the spatial protein organization, which can be crucial for the understanding of many  cellular processes. To overcome this limitation, one can explicitly consider molecular diffusion and use systems of partial differential equations (PDEs) as has been done in the  amyloid-$\beta$ aggregation models \cite{Achdou2013,Franchi2016,Bertsch2016}.  These studies have provided detailed theoretical estimates in terms of boundary conditions and homogenization tools. However, they have restricted the spatial scale to a small three-dimensional region of cerebral tissue and not described intracellular phenomena. Therefore, there is a need for mathematical models for protein aggregation in the PM with a proper spatial description to account for the numerous cellular processes that occur due to heterogeneous protein distribution.

In this work, our primary goal was to investigate the emergence of spatially heterogeneous steady-state profiles of membrane protein aggregates to identify how feedback between cytosolic and membrane components can drive pattern formation on the membrane. 
To this end, we merged the concept of bulk-surface reaction-diffusion systems with the Smoluchowski approach to introduce a new bulk-surface model for membrane protein clustering (\Cref{fig:olig_1}). 
The model equations describe a 
purely diffusive ligand in the cytosol which then undergoes membrane binding, without any cytosolic aggregation. 
The resultant membrane-bound protein can diffuse laterally and also form clusters with different oligomeric sizes.
Finally, the oligomers of maximum size can further recruit more cytosolic proteins, resulting in a positive feedback for the membrane protein aggregates and stabilization of the oligomers \cite{Habchi2018,sarkar2013}.
Following the approach of Ratz and Roger \cite{Ratz2012,Ratz2015}, we then analyzed the model for diffusion-driven instabilities using the classical Turing mechanisms. We found these interactions allow diffusion-driven instabilities and pattern formation in the absence of a sustained localized stimulus.

In what follows, we present the model assumptions and derivation in \Cref{sec:model_dev_sec}, the mathematical analysis including stability analysis in \Cref{sec:math_analysis_sec}, and conclude with numerical simulations (\Cref{sec:num_sim_sec}) and a discussion (\Cref{sec:diss}) about our findings in the context of amyloid-$\beta$ and cluster of other membrane proteins.

\begin{figure}[btp]
\centering
\includegraphics[width=0.85\textwidth]{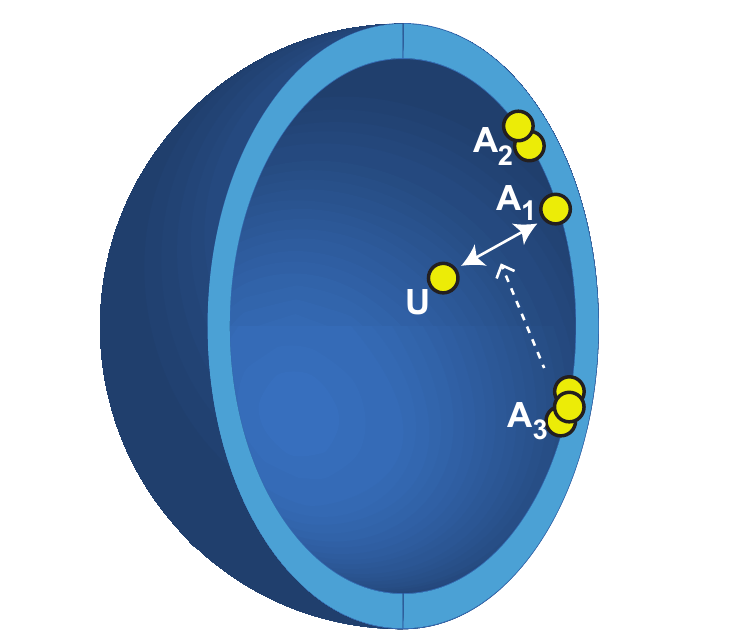}
\caption{\footnotesize A bulk-surface compartmental model for protein aggregation. As proteins approach the surface they can associate and then oligomerize. This oligomerization then drives further membrane association of monomers. Arrows represent a state change of u to a; the dotted line shows the `catalytic' feedback of $a_N$ to u and a.}
\label{fig:olig_1}
\end{figure}

\clearpage

\section{Model Development}
\label{sec:model_dev_sec}

Here we present our bulk-surface reaction-diffusion model for protein aggregation, including feedback. We describe our assumptions (\Cref{assump_subsec}) and the governing equations 
(\Cref{gov_eq_subsec}) in detail. In  \Cref{mass_cons_subsec}, we prove that the total mass of the system is conserved over time, and in \Cref{nom_dim_subsec}, we non-dimensionalize the model. Finally, in section \Cref{sym_red_sec}, we perform the systems reduction when the cytosolic diffusion goes to infinity, following the mathematical approach of Ratz and Roger \cite{Ratz2012,Ratz2015}.

\subsection{Assumptions}
\label{assump_subsec}

In our system, we assume that \textbf{U} represents the volume component, which can freely diffuse in the cytoplasm. Upon binding to the plasma membrane, it forms a surface monomer component $\textbf{A}_1$.  The $\textbf{A}_1$ molecules laterally diffuse in the membrane and form the oligomeric components $\textbf{A}_j$. Here, $j$ denotes the number of $\textbf{A}_1$ molecules in the oligomer, which is at most $N \in \mathbb{N}$.  In terms of chemical reactions, $\ch{$\textbf{U}$   <=>[$f$] $\textbf{A}_1$}$ denotes the binding of the cytosolic component to the plasma membrane with a reaction flux  $f$. The subsequent oligomerization at the membrane is described by 
$$\ch{\textbf{A}_{$j-1$} + \textbf{A}_1 <=> \textbf{A}_{$j$}} \quad \text{for} \quad  j=2,3,...,N.$$
We also assume that the flux term $f$ describes ligand binding/unbinding to the cell surface, where the binding term will be linearly proportional to the concentrations of $U$ in the cytosol and $A_N$ in the plasma membrane. The oligomerization process is modeled as a particular version of the reversible Smoluchowski model for aggregation dynamics \cite{Bentz1981}.  We also assume that the oligomerization process occurs only by monomer attachment in the mass action regime. Moreover, to keep the analysis tractable, we do not consider any cooperativity term such as Hill's function \cite{Changeux1967}: the rate at which the different oligomers are formed is independent of their size.

\subsection{Governing equations}
\label{gov_eq_subsec}

We represent the cellular domain as the bounded region $\Omega$ with smooth boundary $\Gamma = \partial \Omega$. We define the concentrations  $u(x,t): \Omega \times (0,\mathcal{T}] \to \mathbb{R}$ for the volume component and $a_j (x,t) : \Gamma \times (0,\mathcal{T}] \to \mathbb{R}$ for the membrane oligomeric components. The molecular mechanisms underlying membrane protein aggregation and stabilization are quite complex. Therefore, we propose a mathematically tractable \emph{feedback} term to represent this complex mechanism. The flux term is thus defined as  
\begin{equation}
f(u,a_1,a_N) = (k_0 + k_b a_N) u - k_d a_1
\label{eq_feedback}
\end{equation}
for $k_0$, $k_b$ and $k_d$ positive constants, where $k_0$ is the basal binding rate,  $k_b$ is the rate of  $A_N$-dependent binding rate, and $k_d$ is the unbinding rate from the membrane into the cytosol.  Then the governing equations for the spatiotemporal evolution of the different components are given by
\begin{align}
\partial_t u &= D_u \nabla^2  u \label{orig_model_pde:1} \\
\partial_t a_1 &=  D_1 \Delta a_1  + (k_0 + k_b a_N) u - k_d a_1 - 2k_m a^2_1 + 2 k_2 a_2 \nonumber \\
   &  \quad \quad - k_g a_1 \left(\sum_{l=2}^{N-1} a_l\right) + \sum_{j=3}^{N} k_j a_j \label{orig_model_pde:2} \\
\partial_t a_2 &=  D_2 \Delta a_1 + k_m a^2_1 - k_g a_1 a_2 - k_2 a_2 + k_3 a_3  \label{orig_model_pde:3} \\
\partial_t a_j &=  D_j \Delta a_j +  k_g a_1 a_{j-1}  - k_g a_1 a_j  - k_j a_j + k_{j+1} a_{j+1}, \quad j =3,\ldots,N-1 \label{orig_model_pde:4} \\
\partial_t a_N &=  D_N \Delta a_N +  k_g a_1 a_{N-1}  - k_N a_N \label{orig_model_pde:5}
\end{align}

Here, $\nabla^2$ and $\Delta$ represent the Laplace and Laplace-Beltrami  operators, respectively.   The parameter $k_m$ represents the rate at which monomers bind to form dimers. The rate $k_g$ at which the oligomers of size greater than two are formed is assumed to be the same for all oligomerization reactions.  Finally, $k_j$ represent the rates at which the oligomeric components of size $j$ will release a single monomer. The boundary condition for $a$ is periodic since the domain is closed and the boundary condition for $u$ is given by
\begin{equation}
\label{oil_boundary}
 - D_u \left(\mathbf{n}\cdot\nabla u\right ) =  (k_0 + k_b a_N) u - k_d a_1.
\end{equation}

\noindent as a balance of the diffusive flux and the reaction rate at the membrane. All parameters and variables are non-negative real numbers.
 
\subsection{Mass conservation}
\label{mass_cons_subsec}

Let $n_X$ denote the number of molecules of the component $X$. For a closed system, we know that the total number of single molecules must be given by  
\begin{equation*}
    n_U+ n_{A_1}+ 2 n_{A_2} +  ...+ N n_{A_N} 
\end{equation*}
since each $A_j$ oligomer must have exactly $j$ molecules of $A_1$. From this fact, we define the total mass of the system, which accounts for spatial compartments (bulk and surface) and the different molecular size distributions. This is the content of the following.

\begin{prop}
	Let $u$, $a_1$, $a_2$,..., $a_N$ be solutions of \eqref{orig_model_pde:1}--\eqref{oil_boundary}. Then the quantity
	\begin{equation}
	 M(t)  := \int_{\Omega} u(x,t) dx + \sum_{j=1}^{N} \left\{ j \cdot \int_{\Gamma} a_j(x,t) ds  \right\}  
	 \label{mass_cons}
	 \end{equation}
	 represents the total mass of the system and is conserved over time, i.e,  $M(t) = M_0  \quad \forall t \geq 0$. In this case, $M_0$ denotes the initial mass which is given by $M_0 = \int_{\Omega} u(x,0) dx + \sum_{j=1}^{N} \left\{ j \cdot \int_{\Gamma} a_j(x,0) ds  \right\}$
	 \label{mass_cons_prop}
\end{prop}

\begin{proof}
	 By taking the time derivative of $M(t)$, and assuming $u$ and $a_j$ are $\mathcal{C}^2$ solutions for \eqref{orig_model_pde:1}--\eqref{oil_boundary}, we have  
\begin{align*}
	\frac{d}{dt} M   &= \int_{\Omega} \partial_t u dx + \sum_{j=1}^{N} \left\{ j \cdot \int_{\Gamma} \partial_t a_j ds  \right\}
\end{align*}

For the integral $\int_{\Omega} \partial_t u dx$ , we apply the divergence theorem and substitute Eq. \eqref{oil_boundary} to obtain
\begin{align*}
	\int_{\Omega} \partial_t u dx  &= D_u  \int_{\Omega} \nabla^2 u dx  \\
	& =D_u  \int_{\Gamma}  (\nabla u \cdot \textbf{n}) ds   \\		
	& = -	\int_{\Gamma} [(k_0 + k_b a_N) u - k_d a_1] ds				
\end{align*}

For the summation of surface integrals $\sum_{j=1}^{N} \left\{ j \cdot \int_{\Gamma} \partial_t a_j ds  \right\}$, we substitute the governing equations to obtain
\begin{align*}
\hspace*{-0.2cm}
&\sum_{j=1}^{N} \left\{ j \cdot \int_{\Gamma} \partial_t a_j ds  \right\} \\
   &= \int_{\Gamma} \left[D_1 \Delta a_1  + (k_0 + k_b a_N) u - k_d a_1 - 2k_m a^2_1 + 2 k_2 a_2 - k_g a_1 \left(\sum_{l=2}^{N-1} a_l\right) + \sum_{j=3}^{N} k_j a_j  \right] ds  \\
   & \qquad +  \int_{\Gamma} \left[ 2 \cdot D_2 \Delta a_2 + 2 \cdot \left\{ k_m a^2_1 - k_g a_1 a_2 - k_2 a_2 + k_3 a_3 \right\}  \right] ds  \\
   & \qquad +  \sum_{j=3}^{N-1} \int_{\Gamma} \left[ j \cdot D_j \Delta a_j + j \cdot \left\{ k_g a_1 a_{j-1}  - k_g a_1 a_j  - k_j a_j + k_{j+1} a_{j+1}\right\}  \right] ds  \\
   & \qquad + \int_{\Gamma} N \cdot \left[ D_N \Delta a_N ] + N \cdot\left\{  k_g a_1 a_{N-1}  - k_N a_N \right\} \right] ds\\
   &= \sum_{j=1}^{N} j D_j \cdot \int_{\Gamma} \Delta a_j ds + 	\int_{\Gamma} (k_0 + k_b a_N) u - k_d a_1 ds \\
   &= \int_{\Gamma} \left((k_0 + k_b a_N) u - k_d a_1\right) ds
\end{align*}
where the last equality comes from the fact that   $ \int_{\Gamma} \Delta a_j ds  = 0$ as a consequence of  the First Green's Theorem \cite{Oosterom2006}. We therefore have  $$\frac{d}{dt} M =  \int_{\Omega} \partial_t u dx +\sum_{j=1}^{N} \left\{ j \cdot \int_{\Gamma} \partial_t a_j ds  \right\} = 0,$$ from which we conclude that $M(t) = M(0) =: M_0$ for all $t \geq 0$ 
\end{proof}
The mass conservation property for bulk-surface reaction-diffusion models has been established in different contexts \cite{Cusseddu2018,Ratz2012}. However, to the best of our knowledge, it has never been identified in the context of oligomerization reactions.

\subsection{Non-dimensionalization}
\label{nom_dim_subsec}
 
 We introduce a non-dimensional version of the system that allow convenient qualitative interpretation of the system independent of the actual system size, but instead through the ratio of kinetic parameters to the diffusion contributions. We follow the approach in \cite{Ratz2012,Ratz2015} and define $U,A_1,A_2,..., A_N$ be the dimensional concentration quantities where $[U] = \text{mol}/\mu m^3$ and $[A] = \text{mol}/\mu m^2$ for $j=1,\ldots,N$. We also define $L$ and $T$ as the spatial and temporal quantities, where $[L] =\mu m$ and $[T] = s$. We then introduce the non-dimensional variables 
 $$ \hat{u} = \frac{u}{U}, \quad \hat{a}_j = \frac{a_j}{A_j} \quad (j=1,\ldots,N) , \quad \hat{t} = \frac{t}{T}, \quad \text{and} \quad \hat{x} = \frac{x}{L}, $$ 
 which lead to the transformed domains $\hat{\Omega} := \{\xi \in \mathbb{R}^3 \vert\xi L \in \Omega\}$ and  $\hat{\Gamma} = \partial \hat{\Omega}$. By denoting $\hat{\nabla}, \hat{\nabla}^2$, and $\hat{\Delta}$ as the dimensionless gradient, Laplace, and Laplace-Beltrami operators, respectively, and using 
 $$ \nabla = \frac{1}{L} \hat{\nabla}, \quad  \nabla^2 = \frac{1}{L^2}\hat{\nabla}^2, \quad  \Delta = \frac{1}{L^2}\hat{\Delta}, $$
 we can apply the chain rule and rewrite the system \eqref{orig_model_pde:1}--\eqref{orig_model_pde:5} in the form 
\begin{align}
\frac{U}{T}\frac{\partial \hat{u}}{\partial \hat{t}} &= D_u \frac{U}{L^2}  \hat{\nabla}^2 \hat{u},  \quad  \hat{x} \in \hat{\Omega}, \label{nondim_model_pde:1} \\
\frac{A_1}{T}\frac{\partial \hat{a}_1}{\partial \hat{t}} &=  \frac{D_1 A_1}{L^2} \hat{\Delta} \hat{a}_1 + k_d A_1 \Biggl\{ \left( \frac{k_0 U}{k_d A_1} + \frac{k_b A_N U}{k_d A_1} \hat{a}_N\right) \hat{u}  - \hat{a}_1 - 2 \frac{k_m A_1}{k_d} \hat{a}_1^2 + 2\frac{ k_2 A_2}{k_d A_1} \hat{a}_2  \nonumber \\ 
 & \qquad - \hat{a}_1 \left( \sum_{l=2}^{N-1} \frac{A_j k_g}{k_d} \hat{a}_j \right) + \sum_{j=3}^{N} \frac{k_j A_j}{k_d A_1} \hat{a}_j \Biggl\},  \quad  \hat{x} \in \hat{\Gamma},  \label{nondim_model_pde:2} \\
\frac{A_2}{T}\frac{\partial \hat{a}_2}{\partial \hat{t}} &=  \frac{D_2 A_2}{L^2} \hat{\Delta} \hat{a}_2 + k_d A_1 \Biggl( \frac{k_m A_1}{k_d} \hat{a}_1^2 - \frac{k_2 A_2}{k_d A_1} \hat{a}_2 -  \frac{k_g A_2}{k_d} \hat{a}_1\hat{a}_2 + \frac{k_3 A_3}{k_d A_1}\hat{a}_3 \Biggl), \quad  \hat{x} \in \hat{\Gamma}, \label{nondim_model_pde:3} \\ 
\frac{A_j}{T}\frac{\partial \hat{a}_j}{\partial \hat{t}} &=  \frac{D_j A_j}{L^2} \hat{\Delta} \hat{a}_j + k_d A_1 \Biggl(  \frac{k_g A_{j-1}}{k_d} \hat{a}_1\hat{a}_{j-1} - \frac{k_j A_j}{k_d A_1} \hat{a}_j -  \frac{k_g A_j}{k_d} \hat{a}_1\hat{a}_j + \frac{k_{j+1} A_{j+1}}{k_d A_1} \hat{a}_{j+1}\Biggl),  \nonumber \\
 &  \quad  \hat{x} \in \hat{\Gamma}, j=3,\ldots,N  \label{nondim_model_pde:4}\\ 
\frac{A_N}{T}\frac{\partial \hat{a}_N}{\partial \hat{t}} &=  \frac{D_N A_N}{L^2} \hat{\Delta} \hat{a}_N + k_d A_1 \Biggl(  \frac{k_g A_{N-1}}{k_d} \hat{a}_1\hat{a}_{N-1} - \frac{k_N A_N}{k_d A_1} \hat{a}_N\Biggl), \quad  \hat{x} \in \hat{\Gamma}.  \label{nondim_model_pde:5}
\end{align}
The boundary conditions in \eqref{oil_boundary} can be rewritten as
\begin{equation}
 -\frac{D_u U}{L} \left(\mathbf{n}\cdot \hat{\nabla} \hat{u} \right) = k_d A_1 \left\{\left(\frac{k_{0} U}{k_{d} A_{1}}+\frac{k_{b} A_{N} }{k_{d} A_{1}} \hat{a_{N}}\right) \hat{u}-\hat{a_{1}}\right\} 
 \label{nondim_model_pde:6}
\end{equation}

Since $R>0$, we can define the  characteristic concentrations $U$ and $A_j$ by dividing the total mass of the system per total volume and surface area, respectively. We also define the characteristic time with respect to the diffusion $D_1$ of the monomeric component across the cellular surface. Formally, we define 

\begin{equation}
 U = \frac{M_0}{R \cdot |\Gamma|}, \quad A_j =  \frac{M_0}{|\Gamma|} \quad \text{for} \quad  j =1,2, \ldots, N, \quad T = \frac{R^2}{D_1},\quad   L = R
 \label{char_quan}
 \end{equation}
and the dimensionless parameters
\begin{align*}
& \qquad \hat{k}_0 = \frac{k_0 U}{k_d A_1}, \quad \hat{k}_b = \frac{k_b A_N U }{k_d A_1}, \quad \hat{k}_m = \frac{k_m A_1}{k_d}, \quad \hat{k}_j = \frac{k_j A_j}{ k_d A_1} \quad (j=2,\ldots,N),
\\
& \hat{k}_g = \frac{k_g A_j}{k_d} \quad (j=2,\ldots,N-1),   ,\quad \gamma = \frac{k_d R^2}{D_1}, \quad \tilde{D} = \frac{D_u}{D_1},  \quad  d_j = \frac{D_j}{D_1} \quad (j=2,\ldots,N).
\end{align*}
As a result,  \eqref{nondim_model_pde:1}   can be written as 
\begin{equation}
\frac{\partial \hat{u}}{\partial \hat{t}} = \tilde{D}   \hat{\nabla}^2 \hat{u},\label{nondim_v2:1}
\end{equation}
for $\hat{x} \in \hat{\Omega}$ with boundary condition
\begin{equation}
 - \tilde{D} \left(\mathbf{n} \cdot \hat{\nabla} \hat{u}\right) = \gamma  \left\{ \left[ \hat{k}_0 +  \hat{k}_b  \hat{a}_N\right] \hat{u}  - \hat{a}_1\right\}.\label{nondim_v2:2}
 \end{equation}
for $\hat{x} \in \hat{\Gamma}$. Finally, for the surface components, \eqref{nondim_model_pde:2}--\eqref{nondim_model_pde:5} can be written as
\begin{align}
\frac{\partial \hat{a}_1}{\partial \hat{t}} &=  \hat{\Delta} \hat{a}_1 + \gamma \Biggl\{ \left[ \hat{k}_0 +  \hat{k}_b  \hat{a}_N\right] \hat{u}  - \hat{a}_1 - 2 \hat{k}_m\hat{a}_1^2 + 2\hat{k}_2\hat{a}_2 \Biggr.  \nonumber\\
   & \qquad \left. - \hat{k}_g\hat{a}_1 \left( \sum_{l=2}^{N-1} \hat{a}_j \right) + \sum_{j=3}^{N} \hat{k}_j \hat{a}_j \right\}, \label{nondim_v2:3}\\
\frac{\partial \hat{a}_2}{\partial \hat{t}} &=  d_2 \hat{\Delta} \hat{a}_2 + \gamma \left( \hat{k}_m\hat{a}_1^2 -  \hat{k}_2\hat{a}_2 -  \hat{k}_g \hat{a}_1\hat{a}_2 +\hat{k}_2 \hat{a}_3 \right), \label{nondim_v2:4} \\
\frac{\partial \hat{a}_j}{\partial \hat{t}} &=  d_j \hat{\Delta} \hat{a}_j + \gamma \left(\hat{k}_g \hat{a}_1\hat{a}_{j-1} - \hat{k}_j \hat{a}_j -  \hat{k}_g\hat{a}_1\hat{a}_j + \hat{k}_{j+1} \hat{a}_{j+1}\right), \quad j=3,\ldots, N  \label{nondim_v2:5}\\
\frac{\partial \hat{a}_N}{\partial \hat{t}} &=  d_N \hat{\Delta} \hat{a}_N + \gamma \left(\hat{k}_g  \hat{a}_1\hat{a}_{N-1} - \hat{k}_N \hat{a}_N\right). \label{nondim_v2:6}
\end{align}
 
\subsection{System Reduction when $D_u \to \infty$}
\label{sym_red_sec}

We further reduce our system by assuming that the cytosolic diffusion coefficient is much larger than the lateral diffusion on the membrane as has been experimentally observed \cite{postma2004}. The resulting system is uniquely defined on the membrane surface, and the bulk variable $u$ will be represented by an integral operator also called a \emph{non-local functional}. Our approach closely follows the work of Ratz and Roger \cite{Ratz2012,Ratz2015}, though our system can be $N$-dimensional in principle. Formally, if we assume $D_u \to \infty$ and if the initial concentration of $\hat{u}$ is constant over $\hat{\Omega}$, then $\hat{u}$ no longer depends on space and $u = u(t)$.  Therefore, the mass conservation law given by \eqref{mass_cons} implies 
\begin{equation}
 \hat{u}(t) |\hat{\Omega}| + \sum_{j=1}^N \left\{ j \cdot \int_{\hat{\Gamma}} \hat{a}_j ds \right\} = \mathcal{M}_0
 \label{mass_cons_ND}
 \end{equation}
 where $\mathcal{M}_0 = \hat{u}(0) |\hat{\Omega}| + \sum_{j=1}^N \left\{ j \cdot \int_{\hat{\Gamma}} \hat{a}_j(s,0) ds \right\} $ is the total mass of the dimensionless system. We then  define the non-local functional  $$ \mathcal{U}[\hat{a}_1,\hat{a}_2,...,\hat{a}_N](t) : = \frac{1}{|\hat{\Omega}|} \left[ \mathcal{M}_0 -  \sum_{j=1}^N \left\{ j \cdot \int_{\hat{\Gamma}} \hat{a}_j ds\right\}  \right ]$$

as in \cite{Ratz2012,Ratz2015}. Finally, we drop all the hats to obtain the reduced system
\begin{align}
\frac{\partial a_1}{\partial t} &=  \Delta a_1 + \gamma  \mathcal{F}_1(a_1,a_2 \dots ,a_N) \label{nondim:1} \\
\frac{\partial a_j}{\partial t} &=  d_j \Delta a_j + \gamma \mathcal{F}_j(a_1,a_2 \dots ,a_N), \label{nondim:2} \quad j=2,\ldots,N 
\end{align}
where 
\begin{align*}
\mathcal{F}_1 &=  \left[ k_0 +  k_b  a_N\right] \mathcal{U}[a_1,a_2,...,a_N]   - a_1 - 2 k_m a_1^2 + 2k_2a_2  - k_g a_1 \left( \sum_{l=2}^{N-1} a_j \right) + \sum_{j=3}^{N} k_j a_j , \\
\mathcal{F}_2 &=   k_m a_1^2 -  k_2 a_2 -  k_g a_1 a_2 +k_2 a_3 , \\
\mathcal{F}_j &=  k_g a_1 a_{j-1} - k_j a_j -  k_g a_1a_j + k_{j+1} a_{j+1}, \quad j=3,\ldots,N \\ 
\mathcal{F}_N &= k_g  a_1 a_{N-1} - k_N a_N.
\end{align*}

In the next sections, we will provide analytical estimates and numerical simulations to analyze the stability properties of the reduced system  \eqref{nondim:1}--\eqref{nondim:2}.

\section{Mathematical Analysis}
\label{sec:math_analysis_sec}

In this section, we present the mathematical framework for investigating diffusion-driven-instabilities in the system \eqref{nondim:1}--\eqref{nondim:2}. We establish conditions that guarantee the existence and uniqueness of homogeneous steady-states, or the conditions for having multiple steady-states. We also present a characterization for the Jacobian Matrix in the case of homogeneous perturbations. For the non-homogeneous case, the linearization of the non-local functional yields a different Jacobian matrix, and a family of ordinary differential equations is derived to analyze the stability in terms of the eigenfunctions of the Laplace-Beltrami operator. We then apply our framework in the case $N=2$, where we obtain a necessary condition for diffusion-driven instabilities.  Finally, we make some remarks about basic questions concerning the solution theory for the reduced system~\eqref{nondim:1}--\eqref{nondim:2}, as well as the more general model~\eqref{orig_model_pde:1}--\eqref{oil_boundary}.
 
We start with the characterization of the homogeneous steady-states.

\subsection{Homogeneous steady-states}
\label{hom_ss_subsec}

The homogeneous solutions of \eqref{nondim:1}--\eqref{nondim:2} satisfy the ODE system
$$ \frac{d a_j}{dt} = \gamma \mathcal{F}_j(a_1,a_2 \dots ,a_N) \quad j=1,\ldots,N $$ 
and the steady-states in this case are given by $\textbf{a}^* = (a^*_1,a^*_2,a^*_3,\dots,a^*_N)$ such that $$\mathcal{F}_j (\textbf{a}^*) = 0$$ for all $j=1,\ldots,N$. From $ \mathcal{F}_N (\textbf{a}^*) = 0$, we obtain   $ a^*_N  = \frac{k_g a^*_1 a^*_{N-1}}{k_N}$ and proceeding recursively, it is easy to show that 
\begin{equation}
a^*_j  = \frac{k_g a^*_1 a^*_{j-1}}{k_j}  \quad \text{for} \quad j=3,\ldots,N,\quad \text{and} \quad a^*_2  = \frac{k_m (a^*_1)^2}{k_2}.
\label{ajs}
\end{equation}
Hence $a^*_j = C_j(a^*_1)^j$  where $C_1=1$ and
$$ C_j =  \left(\prod_{i=3}^j\frac{k_g}{k_j}  \right)\left(\frac{k_m}{k_2}\right)  \quad \text{for} \quad j=2,\ldots,N.$$ 
Thus from  $ \mathcal{F}_1 (\textbf{a}^*) = 0$, we must have 
\begin{align}
       a^*_1 &= \left[k_0 + k_b \  a^*_N\right]\frac{1}{|\Omega|} \left[ \mathcal{M}_0 -  |\Gamma |\sum_{j=1}^N j \cdot  a^*_j   \right] \nonumber \\
             &=  \left[k_0 + k_b  \  C_N (a^*_1)^N\right]\frac{1}{|\Omega|} \left[ \mathcal{M}_0 -  |\Gamma |\sum_{j=1}^N j \cdot  C_j (a^*_1)^j   \right]. 
      \label{fzero}
\end{align}

By multiplying both sides by $|\Omega|$ and rearranging the $(a^*_1)^j$ terms, we can define the polynomial 
\begin{align}
     \mathcal{P}_N(\alpha)  &=   - k_0 \mathcal{M}_0  \ + \  \left(|\Omega| \ + \  k_0 |\Gamma|\right)\alpha   \ + \  k_0 |\Gamma|\left(  \sum_{j=2}^{N-1} j C_j \alpha^j \right)   \nonumber \\
     & + \   C_N \left( k_0 |\Gamma| N  - \mathcal{M}_0 k_b  \right) \alpha^N \ + \ k_b |\Gamma| C_N \left(\sum_{j=1}^{N} j C_j \alpha^{N+j} \right)
     \label{PN}
\end{align}

     where the roots of $\mathcal{P}_N$ are the steady-state values $a^*_1$. We then observe that the coefficient of $\alpha^N$ is a non-negative number if and only if $$ k_0 |\Gamma| N  - \mathcal{M}_0 k_b \geq 0,$$ which in this case implies that $\mathcal{P}_N(\alpha)$ has a unique positive root and therefore that the system has a unique steady-state. This is the case when $k_b=0$, which means that the largest oligomers do not promote ligand binding in the plasma membrane. On the other hand, if  $ k_0 |\Gamma| N  - \mathcal{M}_0 k_b < 0$, then multiple steady-states could exist.

\subsection{Linear Stability Analysis}
\label{lin_stab_sc}

Linear stability is a traditional concept from the theory of dynamical systems that treats the study of the local behavior near a steady-state solution. The term ``linear'' stands for the analysis of the linear approximation of a nonlinear system, which can be sufficient to determine if a steady-state is stable or unstable. In the case of a system of ODEs, the analysis is carried out by evaluating the eigenvalues of the so-called Jacobian matrix. A similar analysis can be carried in the context of reaction-diffusion systems of PDEs with the analysis of the eigenvectors of the Laplace operator. A major contribution in this field is due to Alan Turing in the classic paper \emph{``The Chemical Basis of Morphogenesis''} \cite{turing1952}. Turing established the notion of diffusion-driven instabilities and was the first to connect this mathematical idea with the formation of spatially heterogeneous patterns. In what follows, we first analyze the homogeneous perturbations of the steady-states by describing the Jacobian matrix of the system. Then we apply the Turing framework and define the conditions for diffusion-driven instabilities in our system \eqref{nondim:1}--\eqref{nondim:2}.

\subsubsection{Homogeneous perturbations}
\label{hom_pert_subsec}

In this section, we investigate the linear stability of the steady-states $\textbf{a}^*$ against spatially homogeneous perturbations, that is in the absence of diffusion. Our study is an N-dimensional version of the approach taken in \cite{Ratz2012,Ratz2015} for a GTPase cylcling model.  We need to compute the eigenvalues $\lambda$ of the Jacobian matrix  $$\mathcal{J} [\textbf{a}^*] = \gamma \left[\frac{\partial \mathcal{F}^*_j}{\partial a_i}\right]_{1 \leq i,j \leq N}$$ for $\mathcal{F}_j$ defined in  \eqref{nondim:1} and \eqref{nondim:2}.  If all the eigenvalues of $\mathcal{J} [\textbf{a}^*]$ have negative real parts, then the steady-state is called linearly stable \cite{strogatz1994}. That means that local perturbations will converge to the steady-state.

On the other hand, if at least one of the eigenvalues has a positive real part, then it is called linearly unstable, which that local perturbations will lead the system away from the steady-state.  The next proposition generally characterizes  $\mathcal{J} [\textbf{a}^*] - \lambda \textbf{I}$. 
\\
\begin{prop}
 The matrix $ \mathcal{J}[\textbf{a}^*] - \lambda \textbf{I}$ can be written in the form
$$  
\begin{bmatrix}
w_0 - \lambda & \textbf{w} \\
\textbf{v}   & H - \lambda \textbf{I}
\end{bmatrix}
$$
where $w_0$ and $\lambda$ are real numbers, $\textbf{w} \in \mathbb{R}^{N-1}$ is a row vector, $\textbf{v} \in \mathbb{R}^{N-1}$ is a column vector, and $H$ is a $(N-1)\times(N-1)$ tridiagonal matrix. 
\end{prop}
    
    \begin{proof}
    We will first calculate $\frac{\partial \mathcal{F}^*_j}{\partial a_i}$ for $i,j =1,2,\ldots,N$ . For $j=1$ we obtain
   \begin{align*}
      \frac{\partial \mathcal{F}^*_1}{\partial a_1} &= - \frac{|\Gamma| \left(k_0 + k_b a^*_N\right) }{|\Omega|} - 1 - 4k_m a^*_1 - \sum_{l=2}^{N-1} k_g a^*_l, \\ 
      \frac{\partial \mathcal{F}^*_1}{\partial a_2} &=  - 2 \frac{|\Gamma| \left(k_0 + k_b a^*_N\right)}{|\Omega|}  - k_g a^*_1 + 2k_2, \\
    \frac{\partial \mathcal{F}^*_1}{\partial a_i} &= - i \frac{|\Gamma| \left(k_0 + k_b a^*_N\right)}{|\Omega|}  - k_g a^*_1 + k_i  \quad \text{for} \quad  i=3,4, \ldots ,N-1, \\
    \frac{\partial \mathcal{F}^*_1}{\partial a_N} &=  - N \frac{|\Gamma|  \left(k_0 + k_b a^*_N\right)}{|\Omega|}+ \frac{k_b}{|\Omega|} \left( \mathcal{M}_0 -  |\Gamma| \sum_{j=1}^N  j \cdot  a^*_j   \right) + k_N.
     \end{align*}
     
     Now for $j=2$, we have
     $$ \frac{\partial \mathcal{F}^*_2}{\partial a_1} = 2 k_m a^*_1 - k_g a^*_2, \quad \frac{\partial \mathcal{F}^*_2}{\partial a_2} =   -  k_g a^*_1 - k_2, \quad \frac{\partial \mathcal{F}^*_2}{\partial a_3} = k_3  \quad \frac{\partial \mathcal{F}^*_2}{\partial a_i} = 0, \quad i=4,5, \ldots, N   $$
    and  for $j=3$ to $j= N-1$, we obtain
       $$\frac{\partial \mathcal{F}^*_j}{\partial a_1} = k_g a^*_{j-1} - k_g a^*_j, \quad \frac{\partial \mathcal{F}^*_j}{\partial a_{j-1}} =  k_g a^*_1, \quad \frac{\partial \mathcal{F}^*_j}{\partial a_{j}} = - k_g a^*_1 - k_j, \quad \frac{\partial \mathcal{F}^*_j}{\partial a_{j+1}} =  k_{j+1}, $$
    and 
    $$ \frac{\partial \mathcal{F}^*_j}{\partial a_{i}} = 0, $$
   otherwise, and finally for $j=N$,
     $$ \frac{\partial \mathcal{F}^*_N}{\partial a_1} = k_g a^*_{N-1},
        \quad \frac{\partial \mathcal{F}^*_N}{\partial a_2} =  k_g a^*_1  ,
     \quad \frac{\partial \mathcal{F}^*_N}{\partial a_N} =  - k_N,
   \quad  \text{and} \quad \frac{\partial \mathcal{F}^*_N}{\partial a_i} = 0 \quad \text{otherwise.}$$
 We then define $\mathcal{J}^*_{i j}:= \gamma \  \frac{\partial \mathcal{F}^*_j}{\partial a_{i}}$, $w_0 := \mathcal{J}^*_{11} - \lambda$, the vectors $\textbf{v}, \textbf{w} \in \mathbb{R}^{N-1}$ such that
 $$ \textbf{v} = \left(  \mathcal{J}^*_{21}  \  \mathcal{J}^*_{31} \cdots \mathcal{J}^*_{N1} \right)^T  \quad \text{and}  \quad  \textbf{w} = \left(  \mathcal{J}^*_{12}  \  \mathcal{J}^*_{13} \cdots \mathcal{J}^*_{1N} \right) $$
  and 
  $$ \setlength{\tabcolsep}{3em}
H = \begin{bmatrix}
  \mathcal{J}^*_{22} - \lambda & \mathcal{J}^*_{23} &0& \cdots &0 & 0 & 0   \\[8pt]
 \mathcal{J}^*_{32} &  \mathcal{J}^*_{33} - \lambda & \mathcal{J}^*_{34} & \cdots & 0 & 0 & 0 \\[8pt]
  \vdots & \vdots & \vdots & \vdots & \vdots  & \vdots & \vdots \\[8pt]
  0 & 0 & 0& \cdots &  \mathcal{J}^*_{N-1N-2} & \mathcal{J}^*_{N-1N-1} - \lambda & \mathcal{J}^*_{N-1N}\\[8pt]
  0 & 0& 0& \cdots & 0 &\mathcal{J}^*_{N N-1} &\mathcal{J}^*_{NN} -\lambda 
\end{bmatrix}_{(N-1) \times (N-1)} $$
 which proves the proposition.
      \end{proof}

\color{black}
\subsubsection{Non-homogeneous perturbations}
\label{non_hom_pert_subsec}

    We now consider a perturbation of the form $\textbf{a}_s = (a_{s,1}, a_{s,2}, ..., a_{s,N})$  for $s \in (-1,1)$ of the homogeneous steady-state $\textbf{a}^*$ in the direction of $\Phi = \left(\varphi_1, \varphi_2, ..., \varphi_N\right)$, for non-homogeneous  $\varphi_j: \Gamma \times (0,T) \to \mathbb{R}$.  Thus for each component, we assume
    $$ a_{s,j} \vert_{s=0} = a^*_j \quad \text{and} \quad \frac{\partial a_{s,j}}{\partial s}\bigg\vert_{s=0}  = \varphi_j, $$
    so we may write the linear approximation $a_{s,j} \approx a^*_j + s \varphi_j$ for  $j=1,\ldots,N$
    $$ a_{s,j} = a^*_j +  s  \  \varphi_j(x,t). $$
    In particular, the linearization of the non-local functional yields $\mathcal{U}[\textbf{a}_s]  \approx  \mathcal{U}[\textbf{a}^*] + s ( \frac{d}{ds}\big\vert_{s=0} \  \mathcal{U}[\textbf{a}_s] )$ where
    \begin{equation}
    \left( \frac{d}{ds}\bigg\vert_{s=0} \mathcal{U}[\textbf{a}_s] \right) = - \sum_{j=1}^N \frac{d}{ds} \bigg\vert_{s=0} \int_{\Gamma} a_{s,j} ds
    = - \sum_{j=1}^N \int_{\Gamma} \varphi_j ds
    \end{equation}
 
 Since we assume that $\varphi_j \in L^2(\Gamma)$ are orthogonal to the constant perturbations, which were analyzed in the previous section, we now consider $$ \int_{\Gamma} \varphi_j ds = 0 \quad \text{for} \quad j=1,\ldots,N, $$
which leads to a linearized system with a constant input $\mathcal{U}[\textbf{a}_s](t) = \mathcal{U}[\textbf{a}^*]$. For the approximation of the component $a_1$, we thus have
\begin{equation}
\partial_t \varphi_1 =   \Delta \varphi_1  + \sum_{j=1}^{N}\tilde{\mathcal{J}}_{1,j}(\textbf{a}^*)  \varphi_j,  \label{a1_diff}
\end{equation}
where 
$$\tilde{\mathcal{J}}_{1,1} =  - \gamma \left\{ 1  + 4k_ma^*_1 + k_g \left(\sum_{l=2}^{N-1} a^*_l\right)\right\}, \quad \tilde{\mathcal{J}}_{1,2}(\textbf{a}^*) =  \gamma \left(2k_2 - k_g a^*_1\right), $$
$$\tilde{\mathcal{J}}_{1,j}(\textbf{a}^*) = \gamma \left(k_j - k_g a^*_j\right), \quad j=3,\ldots, N-1, \quad \text{and} \quad \tilde{\mathcal{J}}_{1,N}(\textbf{a}^*) = \gamma \left(k_N + k_b \ \mathcal{U}[\textbf{a}^*]\right).   $$

The other terms of the Jacobian matrix remain the same as in the case of the homogeneous perturbations, so we omit the explicit calculations. In vector notation, we can then write the linearized system in the form
\begin{equation}
 \partial_t \Phi = \textbf{D} \Delta \Phi + \tilde{\mathcal{J}}(\textbf{a}^*)  \Phi, 
 \label{lindiff}
 \end{equation}
where $\textbf{D}$ is a diagonal matrix such that  $\textbf{D}_{jj} = d_j$ where $d_1 =1$ and  $\tilde{\mathcal{J}}(\textbf{a}^*)$ is the modified Jacobian matrix. We then define $\mathbb{N}_0: = \mathbb{N} \cup \{0\}$ and consider $(\omega_l)_{l \in \mathbb{N}_0} \subset L^2(\Gamma)$, an orthonormal basis of infinitely smooth eigenfunctions of the Laplace-Beltrami operator, i.e, 
$$ -\Delta \omega_l = \eta_l \omega_l.\quad \text{where} \quad 0 = \eta_0 < \eta_1 \leq \eta_2 \leq \cdots.$$
 Then for each $j=1,\ldots,N $ we can express each component $\varphi_j$ as a linear combination 
$$ \varphi_j  =  \alpha_{j0} \  \omega_0  + \sum_{i \in \mathbb{N}}   \alpha_{jl} \  \omega_l   $$
where $\alpha_{jl} = \alpha_{jl}(t)$ for $l \in \mathbb{N}_0$ . Using vector notation, we can define the quantity $\mathcal{A}_l = (\alpha_{1l},\alpha_{2l}, \ldots, \alpha_{Nl})^T $ such that 
$$ \Phi = \mathcal{A}_0 \  \omega_0 + \sum_{l \in \mathbb{N}} \mathcal{A}_i \  \omega_l(x).$$ 

 By substituting the above expansion in \eqref{lindiff}, we obtain the linear ODE system
\begin{equation}
    \frac{d\mathcal{A}_l}{dt}  = \left[- \eta_l \textbf{D} +  \gamma \tilde{\mathcal{J}}(\textbf{a}^*) \right] \mathcal{A}_l \quad \text{for}\quad  l=0,1,2,\ldots
    \label{Ais}
\end{equation}
and diffusion-driven instabilities occur if the above system is unstable for some $l \in \mathbb{N}_0$. This is true when at least one eigenvalue $\lambda$ of the matrix $- \eta_l \textbf{D} +  \tilde{\mathcal{J}}(\textbf{a}^*)$ has a positive real part.  Therefore our target quantity is the so-called \emph{dispersion relation} 
\begin{equation}
h(l): = \max \left( \operatorname{Re}(\lambda(\eta_l)) \right),
\label{disp_rel}
\end{equation}
 where $\operatorname{Re}(z)$ denotes the real part of a complex number $z$.  Finally, the characteristic polynomials $p_l(\lambda): = \det(\lambda \textbf{I} - \gamma \tilde{\mathcal{J}}(\textbf{a}^*) + \eta_l \textbf{D})$ can be written in the form 
$$ p_l(\lambda) = \lambda^N + b_{l,N-1} \lambda^{N-1} + ... +  b_{l,0} $$
where $ b_{l,0} = det(  - \gamma \tilde{\mathcal{J}}(\textbf{a}^*) + \eta_l \textbf{D})$. Therefore, if $b_{l,0} < 0 $ for some $l \in \mathbb{N}$, then $p_l$ has a positive root and therefore $h(l)>0$.

\subsection{Special Case $N=2$: Necessary Conditions for Diffusion-Driven Instabilities}

We now fix $N=2$ and analyze the conditions for diffusion-driven instabilities.  The equations are given by
\begin{align}
\partial_t a_1 &=   \Delta a_1 + \gamma  \left\{ \frac{(k_0 + k_b a_2)}{|\Omega|} \left[ \mathcal{M}_0 -   \int_{\Gamma} \left(a_1 + 2 a_2\right) ds \right ]  -   a_1  - 2k_m a^2_1 + 2 k_2 a_2\right\} \label{red_sys_N2:1}\\
\partial_t a_2 &=  d_2 \Delta a_2 +  \gamma  \left\{k_m a^2_1 - k_2 a_2\right\}. \label{red_sys_N2:2}
\end{align}
 We provide a necessary condition in a special case where the system admits a unique spatially homogeneous steady-state. We prove that the system does not exhibit diffusion-driven instabilities provided that $k_b$ is sufficiently small. 

\begin{thm}
Suppose $k_b \geq 0$ is such that $$k_b \leq \frac{2}{\mathcal{M}_0}\min\left\{ k_0 |\Gamma|, \frac{ d_2 \eta_i |\Omega|}{\gamma}\right\} $$ for all $i \in \mathbb{N}$. Then the system admits a unique steady-state and no diffusion-driven instability exists. 
\label{thm1}
\end{thm}

\begin{proof}
Let $\textbf{a}^* =\left (a^*_1, a^*_2\right)$ be the spatially-homogeneous steady $\textbf{a}^* =\left (a^*_1, a^*_2\right)$, which is obtained when $a^*_2 =  \frac{k_m (a^*_1)^2}{k_2}$ and $a^*_1$ is a solution of  $\mathcal{P}_2(\alpha) = 0$, where  
\begin{align*}
 \mathcal{P}_2(\alpha) &=  - k_0 \mathcal{M}_0  + \left(  |\Omega| + k_0 |\Gamma|\right)\alpha  + \frac{k_m}{k_2} \left( 2 k_0 |\Gamma|   - \mathcal{M}_0 k_b  \right) \alpha^2 \\
& \qquad + k_b |\Gamma|\frac{k_m}{k_2} \alpha^{3} + 2 \left(\frac{k_m}{k_2}\right)^2  k_b |\Gamma| \alpha^4.
\end{align*}

Now, since  $k_b \leq  \frac{2k_0 |\Gamma|}{\mathcal{M}_0}$, we have $ 2 k_0 |\Gamma|   - \mathcal{M}_0 k_b \geq 0 $, and therefore the intermediate value theorem ensures that the system  admits a unique positive steady-state. The Jacobian matrix with respect to homogeneous perturbations is then given by  
 $$ \mathcal{J}[\textbf{a}^*] = \gamma \left[
\begin{array}{cc}

 -1-4 a^*_1 k_m- \frac{ |\Gamma| }{|\Omega|}\left(\frac{k_b k_m (a^*_1)^2}{k_2}+k_0\right) 

& 2 a^*_1 k_m \\

2 k_2+ \frac{k_b}{|\Omega|} \left[\mathcal{M}_0 - \left(\frac{2 k_m (a^*_1)^2}{k_2}+a^*_1\right) |\Gamma| \right] -\frac{2|\Gamma|}{|\Omega|} \left(\frac{k_b k_m (a^*_1)^2}{k_2}+k_0\right)

& -k_2

\end{array}
\right]$$
   
with a second-order characteristic polynomial $p(\lambda) = det(\lambda \textbf{I} -  \mathcal{J}[\textbf{a}^*])$ given by $p(\lambda) =  \lambda^2  + b \lambda  + c $, where $$b =  \gamma \left(\frac{(a^*_1)^2 |\Gamma|  k_b k_m}{k_2 |\Omega| }+4 a^*_1 k_m+\frac{|\Gamma|  k_0}{|\Omega| }+k_2+1\right) >0 $$
and 
$$ c =  \gamma^2 \left( \frac{8 (a^*_1)^3 |\Gamma|  k_b k_m^2}{k_2 |\Omega| }+ \frac{3 (a^*_1)^2 |\Gamma|  k_b k_m}{|\Omega| }+ \frac{2 a^*_1 k_m}{|\Omega|} \left(2 |\Gamma|k_0 - k_b \mathcal{M}_0 \right) + \frac{|\Gamma|  k_0 k_2}{|\Omega| }+ k_2  \right) $$ is also positive because  $ 2 k_0 |\Gamma|   - \mathcal{M}_0 k_b \geq 0 $. From that we conclude that both eigenvalues 
$$ \lambda  = \frac{-b \pm \sqrt{b^2 - 4c}}{2} $$
must have real negative parts and therefore the steady-states are linearly stable. We then perform a similar argument for  non-homogeneous perturbations. From \eqref{a1_diff}, we obtain the modified Jacobian matrix  $\tilde{\mathcal{J}}(\textbf{a}^*)$, and for a given $l \in \mathbb{N}_0$, we have
 \begin{equation*}
 \tilde{\mathcal{J}}(\textbf{a}^*) - \eta_l \textbf{D} =  
\begin{bmatrix}
  - \gamma(1+4 a^*_1 k_m) - \eta_l & \quad  \gamma \left\{ 2 k_2 + \frac{k_b}{|\Omega|}  \left[\mathcal{M}_0 - |\Gamma| \left(\frac{2 k_m (a^*_1)^2}{k_2}+a^*_1\right) \right] \right\}   \\[15pt]  
 2 \gamma  \ a^*_1  \ k_m &\quad  - \gamma \ k_2 - d_2 \eta_l  \\
\end{bmatrix}
\end{equation*}
with characteristic polynomials  $p_l(\lambda): = \det(\lambda \textbf{I} - \gamma \tilde{\mathcal{J}}(\textbf{a}^*) + \eta_l \textbf{D})$ given by the quadratics $p_l(\lambda) =  \lambda^2  + b_l \lambda  + c_l$, where 
$$b_l =  \gamma  \  (4 a^*_1 k_m+k_2+1) + \eta_l (d_2 + 1) >0 $$
and 
\begin{align*}
c_l &= \left[d_2 \eta_l^2+\gamma  \eta_l (d_2+k_2)+\gamma ^2 k_2 \right] +  2 \ \gamma  k_m  \ a^*_1  \left(2 d_2 \eta_l - \frac{\gamma  k_b \mathcal{M}_0}{|\Omega|}\right) \\
& \qquad
    +  \frac{2 \gamma^2 (a^*_1)^2 |\Gamma| k_b k_m  }{|\Omega| }\left(1  + \frac{2  a^*_1  k_m }{k_2} \right)
\end{align*}

In the case  $l \in \mathbb{N}$, the $c_l$ terms are also positive, since we assume  $$  k_b \leq \frac{2  d_2 \eta_l |\Omega|}{ \gamma \mathcal{M}_0} \iff  2 d_2 \eta_l - \frac{\gamma  k_b \mathcal{M}_0}{|\Omega|} \geq 0  \quad \forall i \in \mathbb{N},$$ and in this case the $p_l(\lambda)$ have no roots with positive real parts. In the  case $l=0$, we know that $\eta_0 = 0$, but the modified matrix $\tilde{\mathcal{J}}(\textbf{a}^*)$  yields a different linearized system. Thus we have to analyze the stability of the linear equation given in \eqref{Ais}  in the case where $l=0$, i.e, 
 $$\frac{d\mathcal{A}_0}{dt}  = \left[ \gamma \tilde{\mathcal{J}}(\textbf{a}^*) \right] \mathcal{A}_0 $$
with characteristic polynomial $p_0(\lambda) = \lambda^2 + b_0 \lambda + c_0$ where
$$b_0 =  \gamma  \  (4 a^*_1 k_m+k_2+1) >0 $$ and 
\begin{align}
c_0 &= \gamma^2 \left\{ k_2  -  2k_m  \ a^*_1  \left( \frac{  k_b \mathcal{M}_0}{|\Omega|}\right) 
    +  \frac{2 (a^*_1)^2 |\Gamma| k_b k_m  }{|\Omega| }\left(1  + \frac{2  a^*_1  k_m }{k_2} \right)\right\} \nonumber \\
    &=  \gamma^2 \left\{ k_2  - \frac{2 k_m k_b a^*_1}{|\Omega|}\left[ \mathcal{M}_0  -   |\Gamma| \left(a^*_1 +   2  \frac{ (a^*_1)^2  k_m }{k_2} \right) \right] \right\}.
    \label{ci}
\end{align}
    
 We now verify that $c_0 \geq 0$. In fact, from \eqref{fzero} when $N=2$, we obtain $$ \frac{1}{|\Omega|}\left[ \mathcal{M}_0  -   |\Gamma| \left(a^*_1 +   2  \frac{ (a^*_1)^2  k_m }{k_2} \right) \right] = \frac{k_2 a^*_1}{\left[k_0 k_2 +  k_b  (a^*_1)^2  k_m \right] },$$ 
 and therefore by substituting the above equation on \eqref{ci} and using that $a^*_2 = \frac{k_m}{k_2} (a^*_1)^2$, we obtain 
 $$ c_0  = \gamma^2   k_2 \left\{ 1  -   2 \ \frac{ a^*_2 }{\left[\frac{k_0 }{ k_b } +  a^*_2   \right] }  \right\}.$$ 
    
 Finally, the hypothesis gives us $\frac{k_0}{k_b} \geq \frac{\mathcal{M}_0}{2 |\Gamma|}$ and by using $ \mathcal{M}_0 - 2 |\Gamma| a^*_2 \geq 0 $ (total mass of $a_2$  at steady-state does not exceeds the total mass of the system), we  obtain $ \frac{\mathcal{M}_0}{2 |\Gamma|} \geq a^*_2 $ and therefore 
 $$  \frac{ a^*_2 }{\left[\frac{k_0 }{ k_b } +  a^*_2   \right]} \leq \frac{1}{2} $$
from which we conclude that $c_0 \geq 0$. Therefore, the steady-state is stable against non-constant perturbations.

\end{proof}

\subsection{Existence of Solutions to the General and Reduced Models}

In the sections above, we developed a stability analysis of the reduced model \eqref{nondim:1}--\eqref{nondim:2}, motivated by Turing in \cite{turing1952}, and essentially following the technical approach in \cite{Ratz2012, Ratz2015}.
Understanding the stability of solutions (if any exist) with respect to parameters in the reduced model is part of a larger set of basic mathematical questions regarding the solutions to the reduced system \eqref{nondim:1}--\eqref{nondim:2}, as well as to the more general system \eqref{orig_model_pde:1}--\eqref{oil_boundary}.
Knowing when these types of models are \emph{well-posed} (solutions exist, are unique, and depend continuously on problem data) is critical for drawing scientific conclusions from mathematical estimates and numerical simulations of solutions.
When models of critical phenomena correctly capture solution instability and even multiplicity, knowing when and how many solutions are permitted by the model is also crucial.
In the context of dynamical problems, a standard approach to developing a solution theory is to prove that solutions exist for small time (\emph{local existence}), and then endeavor to extend such results to large time (\emph{global existence}).
In developing existence results one generally first considers
\emph{small data} (initial and boundary conditions being small perturbations of zero), with again the hope of extending such results to \emph{large data} (essentially no restriction on initial or boundary data, other than reasonable smoothness assumptions).
Establishing that solutions are \emph{unique} typically involves exploiting some inherent mathematical structure in the underlying equations.

In~\cite{Sharma2016}, an analysis of a reaction-diffusion system similar to our general model~\eqref{orig_model_pde:1}--\eqref{oil_boundary} is developed; see also~\cite{Jero83} for related results.
The main results in~\cite{Sharma2016} include a global-in-time, large data, existence result based on \emph{a priori} estimates for a linearized model and compactness techniques, together with a fixed-point argument using on a variant of  the Leray-Schauder Theorem (Schaefer's Theorem).
They also obtain a uniqueness result using a Gronwall-type inequality.
Their results require growth conditions on the nonlinearities appearing in the system (a local type of Lipschitz property, together with pointwise control of individual components in the solution).
These conditions appear to be too restrictive to cover the nonlinearities arising here in~\eqref{orig_model_pde:1}--\eqref{oil_boundary}, which include quadratic and more rapidly growing terms. 
Nevertheless, we expect that results similar to those in~\cite{Sharma2016} can be shown to hold for~\eqref{orig_model_pde:1}--\eqref{oil_boundary} using similar arguments.

\section{Numerical Simulations}
\label{sec:num_sim_sec}

 We perform numerical simulations to complete our mathematical analysis.  We start by analyzing the parameter regions of bistability (\Cref{bist_sec}). Then we investigate whether the stable steady-states become linearly unstable under non-homogeneous perturbations (\Cref{inst_sec}). From the linear instability analysis, we obtain the  \emph{single-patch} non-homogeneous steady-state (\Cref{patt_sec}). Finally, we study the temporal dynamics of pattern formation (\Cref{temp_evol}) and the single-patch dependence on the cell radius (\Cref{change_radius}). The numerical simulations were implemented in Matlab R2018a and Comsol Multiphysics 5.4.

\subsection{Bistability under homogeneous perturbations}
\label{bist_sec}

We begin by computing the homogeneous steady-states $\textbf{a}^*$ and the corresponding eigenvalues of the Jacobian matrix $\mathcal{J} [\textbf{a}^*]$ under homogeneous perturbations (cf.~\Cref{hom_pert_subsec}). We then explore the parameter regions of bistability where the system admits three steady-states, two of them stable and one unstable. In the case $N=2$, we obtain regions of bistability by change the basal binding rate $k_0$ and the $A_2$-dependent binding rate $k_b$ (\Cref{fig:bistability}). For $k_0=0.015$, three steady-state values for $a^*_1$ emerge depending on $k_b$ (\Cref{fig:bistability} (A)). When $k_0$ also changes, we obtain both a \emph{bistability region} (dark-gray) and a \emph{single steady-state region} (light-gray) (\Cref{fig:bistability} (B)). A colored vertical line represents the region from \Cref{fig:bistability} (A). Other parameter choices also lead to bistability regions (see \Cref{fig:SFig_1} (A) for $N=2$ and \Cref{fig:SFig_2} (A) for $N=3$ in the ESM).

\begin{figure}[tbp]
\centering
\includegraphics[width=0.98\textwidth]{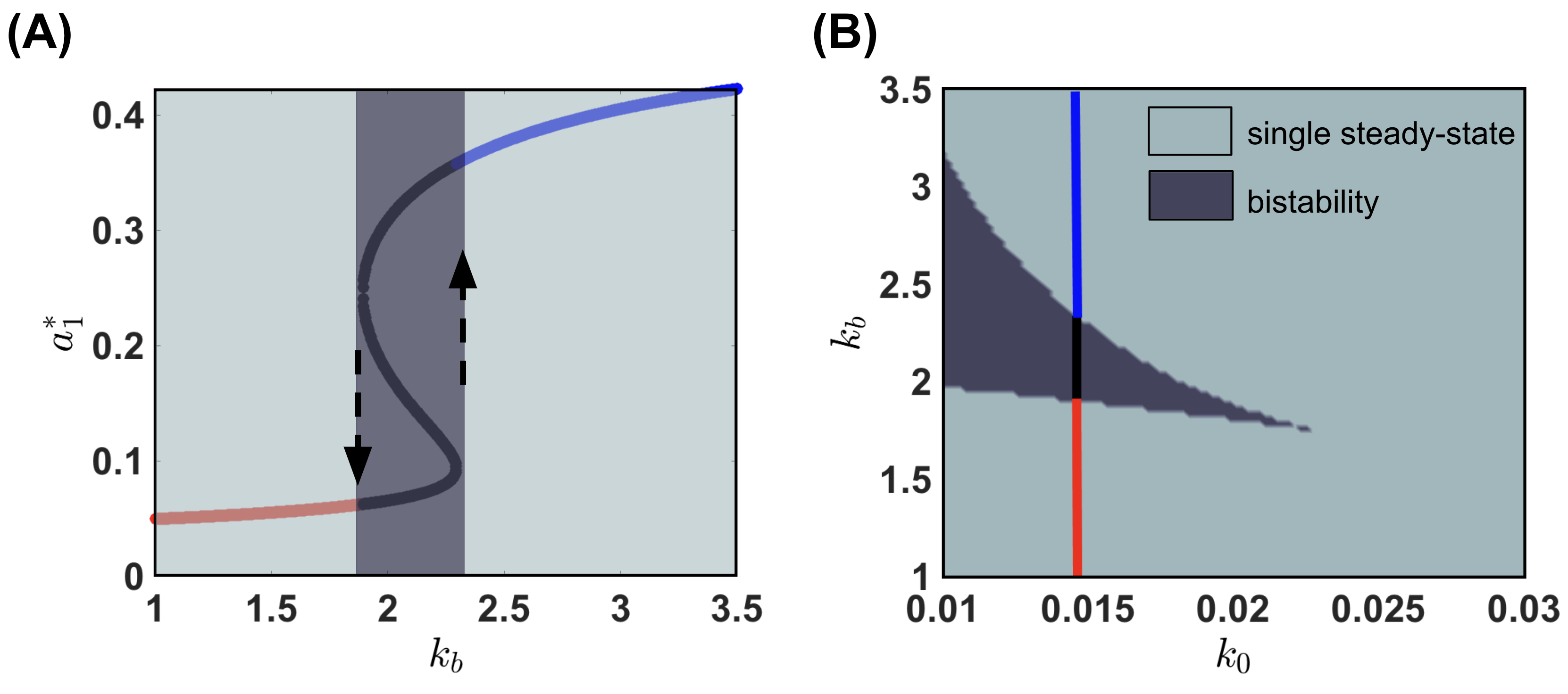}
\caption{\footnotesize\textbf{Steady-states and Parameter Regions for Bistability ($N=2$).} (A) The value of $k_0=0.015$ is fixed, while $k_b$ ranges from 1 to 3.5. We then compute the steady-states, which are the solution of \eqref{fzero}. The single steady-state branches are shown in red and blue, respectively, while the bistable branch is shown in black. The dark-grey rectangle illustrates the emergence of bistability, and the dashed black arrows indicate the stable steady-states. (B) Bistability region for $k_0 \in [0.01,0.03]$ with $k_0=0.015$ marked. The dark gray region contains the $k_b$ values for which the system admits a bistability region. The single steady-state regions are indicated in light-gray.}
\label{fig:bistability}
\end{figure}

\subsection{Linear instability under non-homogeneous perturbations}
\label{inst_sec}

In this section, we numerically investigate which parameter values promote linear instability under non-homogeneous perturbations. We fix an eigenmode index $l \geq 1$ to explore diffusion-driven instabilities, and let $\textbf{a}^*$ be a stable steady-state under homogeneous perturbations. We can thus compute the dispersion relation $h(l)$ \eqref{disp_rel} defined in  \Cref{non_hom_pert_subsec}. If  $h(l)<0$, the steady-state remains stable in the direction of the chosen eigenmode. In this case, the analysis is  inconclusive, since we would also need to determine the stability for the other eigenmodes. If  $h(l)>0$, the steady-state becomes unstable for the chosen eigenmode, and this is sufficient to ensure a diffusion-driven instability \cite{Smith2018}. The case $h(l)=0$ is usually requires higher-order analysis, so we will not consider it in the context of linear stability. 

Given a fixed eigenmode index $l$, we can then divide the parameter space into four regions. We will call them Regions $0$, $1$, $2$, and $3$, where the numbers reflect the exact number of unstable steady-states. More precisely, we define:

\begin{itemize}
    \item \textbf{Region 0}:  The single steady-state region where $h(l)<0$;  There are no unstable steady-states.
    \item \textbf{Region 1}: The single steady-state region where $h(l)>0$;  There is only \underline{one} unstable steady-state. 
    \item \textbf{Region 2}:  The bistability region where $h(l)>0$ for only one stable steady-state; a total of \underline{two} unstable steady-states.
     \item \textbf{Region 3}: The bistability region where $h(l)>0$ for both stable steady-state; a total of \underline{three} unstable steady-states.
\end{itemize}

 Once we find the bistability and single steady-state regions, we can divide the same parameter space into Regions 0, 1, 2, and 3, by computing the number of unstable steady-states (\Cref{fig:instability}(A) for $N=2$ and (B) for $N=3$). The stability analysis in Region 0 is more subtle and requires further analysis since the stability criterion needs to be fulfilled for all eigenmodes. However, at least for $N=2$, \Cref{thm1} ensures that the system remains stable for sufficiently small $k_b$, which appears to be consistent with the numerical predictions. For higher $k_b$ values, the instabilities emerge in the bistability region (Regions $2$ and $3$) and also in the single steady-state regions Regions ($0$ and $1$).  We obtain a similar result for $N = 3$ (\Cref{fig:instability} (B) ). However, it should be noticed that the $k_b$ values promoting linear instabilities are higher (see $y$-axis ranging from $4$ to $14$) compared to with the case $N=2$.  Regions 0, 1, 2, and 3 can be found with other parameter choices (see \Cref{fig:SFig_1} (B) for $N=2$ and \Cref{fig:SFig_2} (B) for $N=3$ in the ESM).

\begin{figure}[tbp]
\centering
\includegraphics[width=0.90\textwidth]{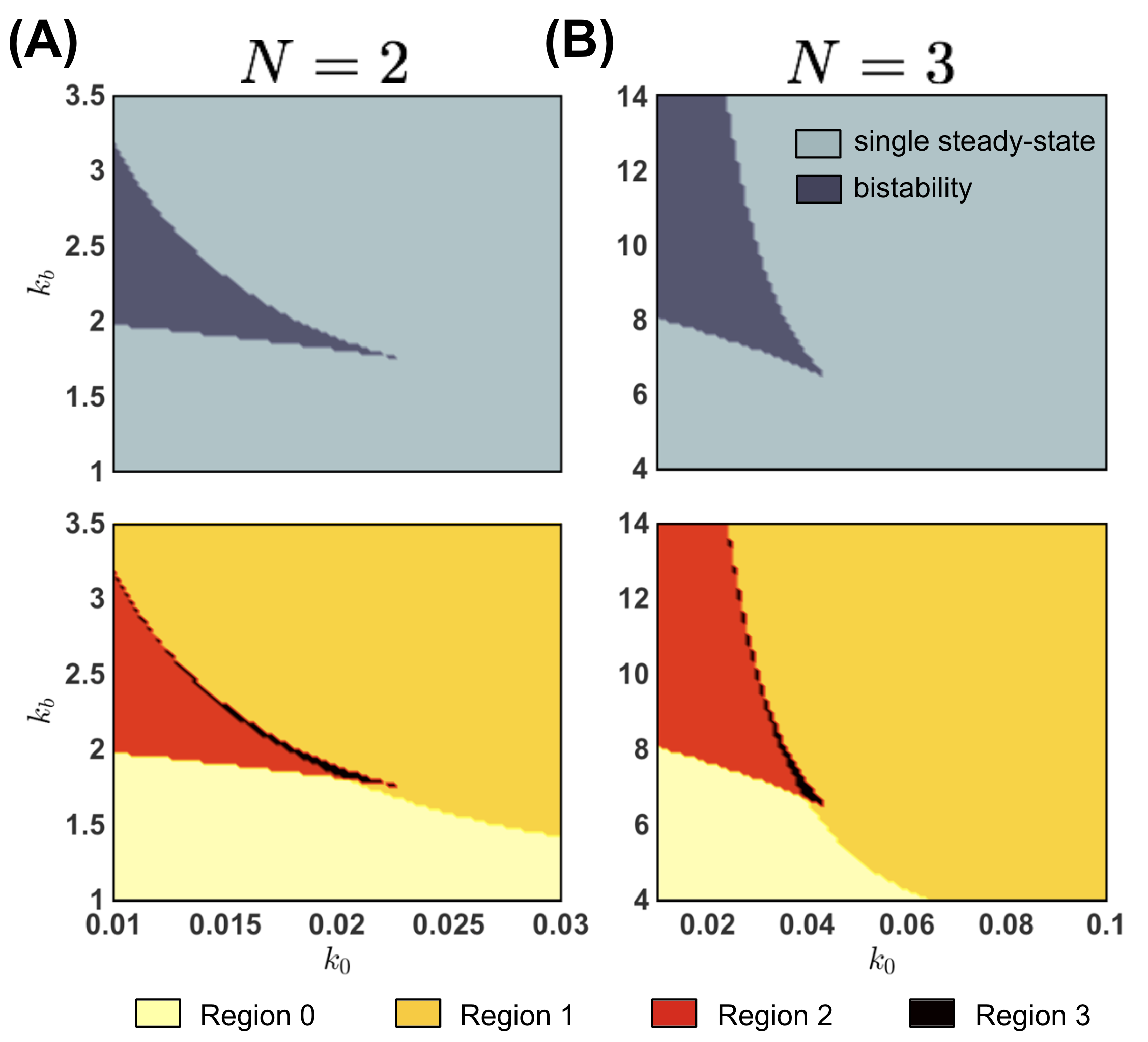}
\caption{\footnotesize\textbf{Parameter Regions of Bistability and Linear Instability ($N=2$ and $N=3$)} We scan the reaction rates for different parameter values. In the top, the parameter regions in the $k_0 \times k_b$ plane where the system exhibits bistability under homogeneous perturbations. In the bottom, Regions 0, 1, 2, and 3 divide the $k_0 \times k_b$ plane according to the number of unstable steady-states under non-homogeneous perturbations for the eigenmode $l=1$ (see text for details).  (A)  $N=2$, $d_2=0.1$,  $\gamma= 1000$, and $j = 1$. (B) $N=3$, $d_2=d_3=0.1$,  $\gamma= 1000$, and $j = 1$. The $k_b$ values that promote linear instability are significantly higher for $N=3$ compared to the case $N=2$.}
\label{fig:instability}
\end{figure}

The Region 1 where a stable steady-state becomes unstable are known as  \emph{Turing-type instability regions} \cite{turing1952,murray1993}, where the system may converge to a spatially non-homogeneous steady-state.  We analyze these region when we increase both $k_0$ and $k_b$ ranges for different values of the diffusion coefficient $d_2$ (\Cref{fig:changes} (A)). As $d_2$ decreases, Region 1 increases, which illustrates how the system becomes unstable as the discrepancies between diffusion become higher. A similar phenomenon occurs as we increase the dimensionless parameter $\gamma$ (\Cref{fig:SFig_3}). On the other hand,  as the eigenmode index $l$ increases, Region 1 significantly decreases  (\Cref{fig:changes} (B)). We exhibit the results for
$l = 2$, $l = 6$, and $l = 8$. Such a decrease allows us to explore the instability of the system by considering $(k_0,k_b)$ on Region 1 only for the first eigenmode $(l=1)$.

\begin{figure}[tbp]
\centering
\includegraphics[width=0.98\textwidth]{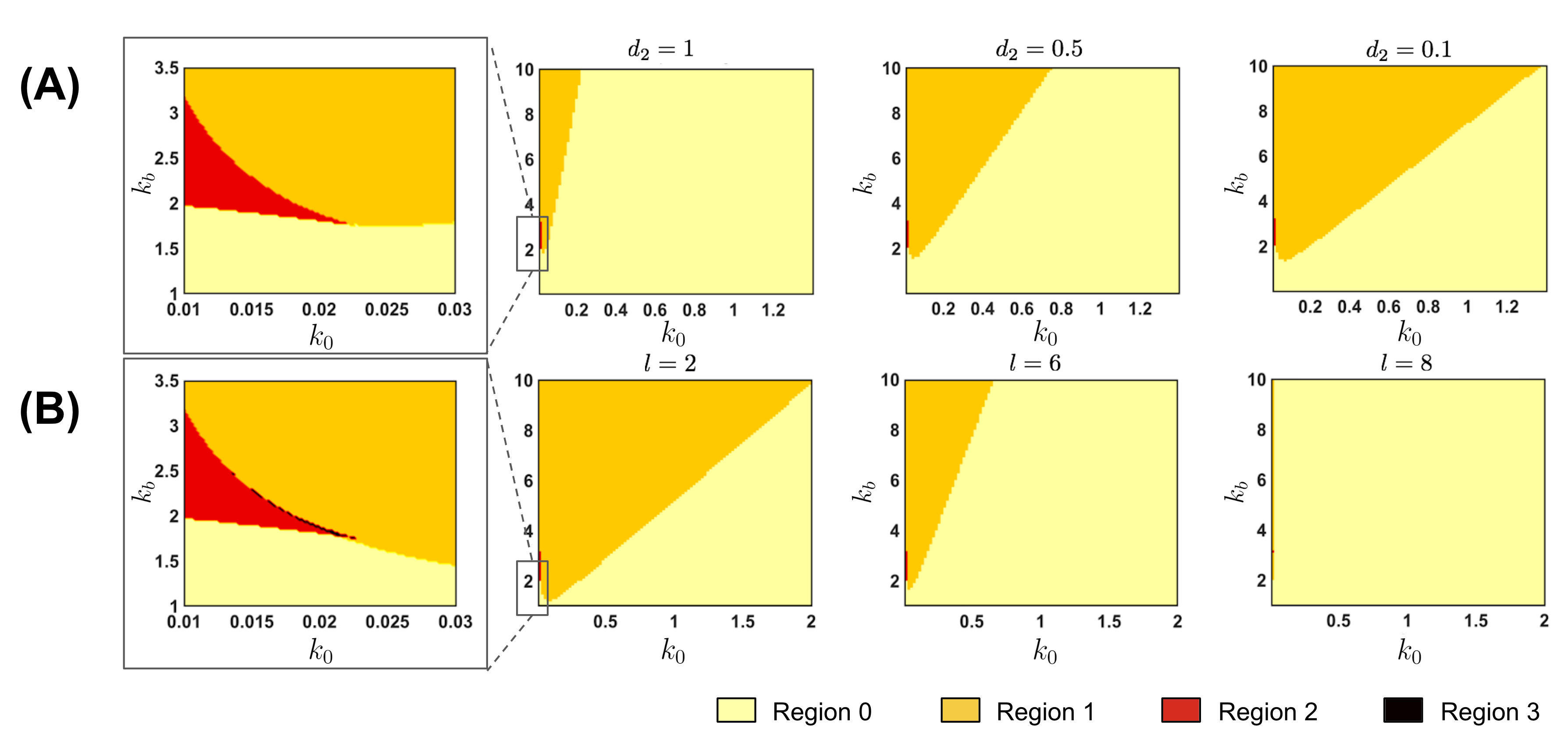}
\caption{\footnotesize\textbf{Changing the diffusion coefficient and the eigenmode of the Laplace-Beltrami operator for $N=2$.}  (A) For $d_2=1$, we show a zoomed plot of the interface of the Regions 1 and 2. Most of  the $(k_0, k_b)$ in the rectangle $[0.01,1.4] \times [1,10]$ belongs to the Region 0, where  the system is stable  under non-homogeneous perturbations. However, by decreasing $d_2$ to $0.5$ and further to $0.1$, the Region 1 (in orange) significanly increases, which means that the system exhibits a larger instability region for lower $d_2$ values. In this figure, we fix $\gamma=10$ and $j=1$ as the eigenmode index. 
(B)  Linear instability Region 1 for eigenmode index values $l=2$, $6$, and $8$. For $l=2$, the system is unstable under non-homogeneous perturbations for most $(k_0, k_b)$ values above the diagonal of the rectangle $[0.01,2] \times [1,10]$. As $l$ increases, Region 1 (in orange) significantly decreases. Therefore, we can analyze the instability of the system by exploring only the first eigenmode, since Region 1 does not expand as $l$ increases.   In this figure, we  fix  $\gamma =100$ and  $d_2=0.1$.
}
\label{fig:changes}

\end{figure}

\subsection{The emergence of the single-patch non-homogeneous steady-state}
\label{patt_sec}

In this section, we investigate the spatio-temporal behavior of our system by numerically integrating the dimensionless equations. We consider a spherical domain of radius $R=1$ and, as in the previous sections, we fix $N=2$ or $N=3$. We avoid solving the surface system \eqref{nondim:1}--\eqref{nondim:2} due to the numerical complexity of the non-local functional. Instead, we solve the dimensionless bulk-surface equations \eqref{nondim_v2:1}--\eqref{nondim_v2:6} (dropping all the hats) for an extremely high cytosolic diffusion ($\tilde{D}=10^8$) on \eqref{nondim_v2:1}. In this way, our resulting system can be seen as an approximation of the reduced system when $\tilde{D} \to \infty$.  We randomly perturbed the homogeneous steady-states by considering a small number $\varepsilon>0$ as the perturbation magnitude and a family $\{\xi(x)\}_{x \in \Gamma}$ of independent random variables uniformly distributed between $-\varepsilon$ and $\varepsilon$. In the case where $N=2$, we define the surface initial conditions

\begin{equation}
a_1(x,0):= a^*_1 + \varepsilon \ \xi(x)  \quad \text{and} \quad  a_2(x,0):= a^*_2 - \frac{1}{2} \varepsilon \ \xi(x) 
\label{rand_N2}
\end{equation}
where the $\frac{1}{2}$ accounts for mass conservation (see \eqref{mass_cons_ND}). For the volume component, we define $u(x,0) := u^*$, where $u^* = \frac{1}{|\Omega|}\left[ \mathcal{M}_0 - |\Gamma| (a^*_1 + 2 a^*_2)\right] $
also because of the mass conservation property. For $N=3$ we define $a^*_1$ as in \eqref{rand_N2}, $a_j(x,0):= a^*_j - \frac{1}{5} \varepsilon \ \xi(x) $ for $j=2$ and $j=3$, and $u^* = \frac{1}{|\Omega|}\left[ \mathcal{M}_0 - |\Gamma| (a^*_1 + 2 a^*_2 + 3 a^*_3)\right]$. 

\begin{remark} The element $(a^*_1,a^*_2,...,a^*_N)$ is a homogeneous steady-state of the system \eqref{nondim:1}--\eqref{nondim:2} if and only if $(a^*_1,a^*_2,...,a^*_N,u^*)$ is a homogeneous steady-state of the system \eqref{nondim_v2:1}--\eqref{nondim_v2:6} provided that 
$$ u^* = \frac{1}{|\Omega|}\left[ \mathcal{M}_0 - |\Gamma| (a^*_1 + 2 a^*_2 + ..+ N a^*_N)\right] $$
\end{remark}

By the remark above, we can obtain the steady-states given the parameter choice in the reduced system \eqref{nondim:1}--\eqref{nondim:2}. Then we can numerically integrate the bulk-surface PDE system   \eqref{nondim_v2:1}--\eqref{nondim_v2:6} using the perturbation scheme described above. In order to associate the parameter regions that lead to instabilities with the formation of spatial patterns, we select four $(k_0,k_b)$ values in the four Regions 0, 1, 2, and 3 (\Cref{fig:patternform}). We fix $N=2$ and the eigenmode index $l=1$. For each  choice of $(k_0,k_b)$ , we integrate the system \eqref{nondim_v2:1}--\eqref{nondim_v2:6} to its final state by perturbing a homogeneous steady-state. We then plot the result for the $a_1$ component and visually inspect the results. For $(k_0,k_b)$ in Regions 1, 2, and 3,  (colored in orange, red or black, respectively), a \emph{single-patch} spatially heterogeneous steady-state emerges. On the other hand, when $(k_0,k_b)$ belong to Region 0, in which the system is stable for the eigenmode index $l=1$, the system converges to its homogeneous steady-state. This result indicates that the single-patch pattern is consistent across parameter choices in Regions 1, 2, and 3, once it remains unchanged in its circular shape and gradient of concentrations. \Cref{fig:SFig_4}
in the ESM  shows a similar result in the case $N=3$.

\begin{figure}[tbp]
\centering
\includegraphics[width=0.98\textwidth]{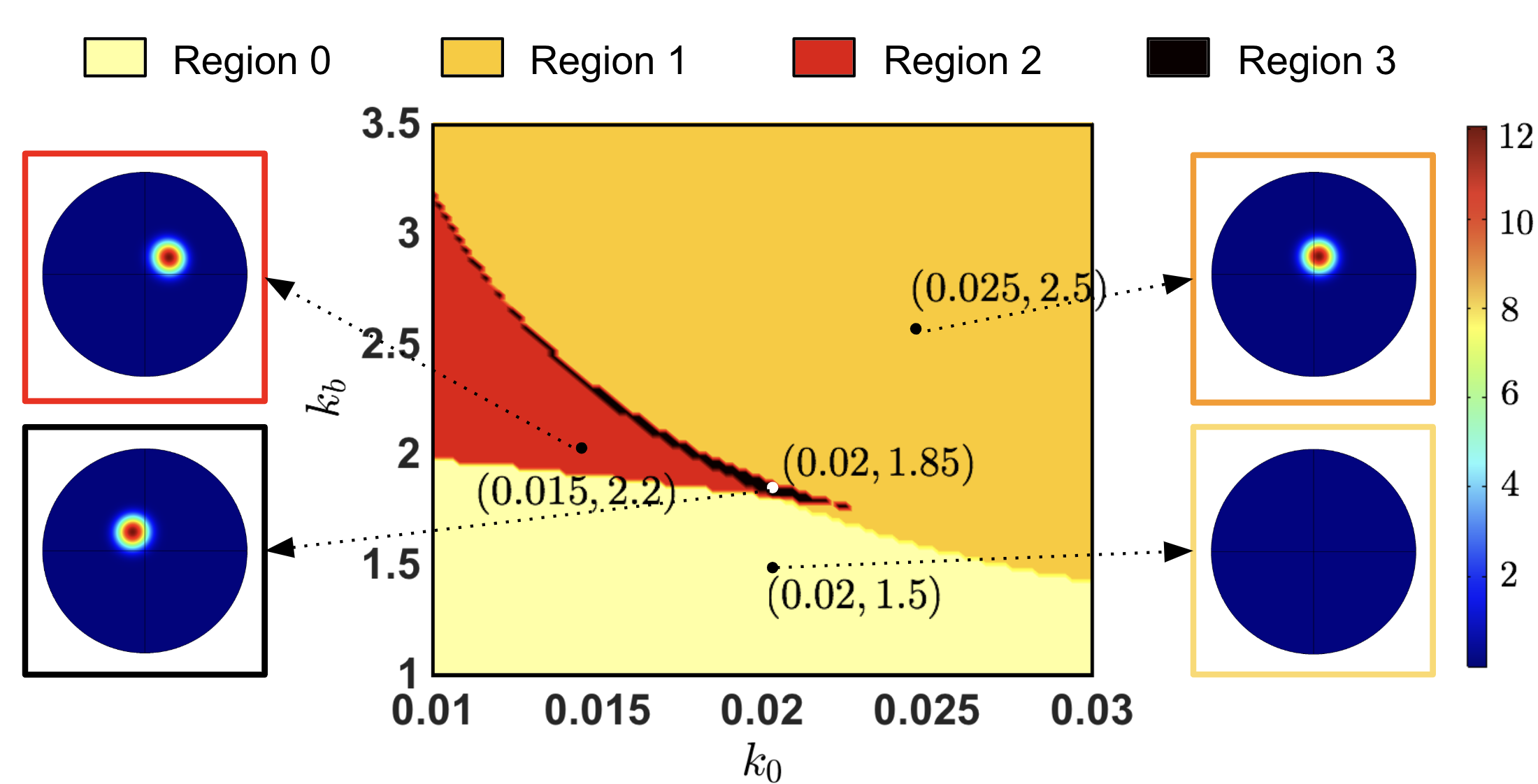}
\caption{\footnotesize \textbf{Linear Instability and Pattern Formation ($N=2$)}.  We exhibit the stability analysis colormap  for eigenmode index $l=1$ and the final spatial profile of the $a_1$ component. We consider four $(k_0,k_b)$ values from Regions 0, 1, 2, and 3, which are colored in light-yellow, orange, red or black, respectively. For Regions 1, 2, and 3, we observe the emergence of a single-patch spatially heterogeneous steady-state which is consistent across parameter regions in terms of its circular shape and concentration gradient. For Region 0, we do not observe a pattern formation for this particular eigenmode.  In this figure, $d_2=0.1$, $\gamma= 1000$, $k_m = k_2 = 1$. steady-state values. Region 0: $a^*_1=0.0812$, $a^*_2=0.0066$, $u^* = 2.7168$. Region 1: $a^*_1 =  0.3817$ , $a^*_2 = 0.1457$ , $u^* = 0.9806$. Region 2: $a^*_1=0.2759$, $a^*_2 = 0.0761$, $u^* = 1.7155$. Region 3: $a^*_1= 0.1107$, $a^*_2 = 0.0123$, $u^*  = 2.5942$}
\label{fig:patternform}
\end{figure}

\subsection{Temporal evolution and pattern formation}
\label{temp_evol}

In this section, we further investigate the temporal evolution of the system.  We consider $N=2$  and $(k_0,k_b) = (0.025, 2.5)$ that belongs to Region 1 (see \Cref{fig:patternform}).  We then observe the spatial distribution of $a_1$ for different times (\Cref{fig:tempevol} (A)). At $t=0$, We apply a random perturbation of magnitude $\varepsilon = 10^{-10}$ around the unique homogeneous steady-state that is unstable under non-homogeneous perturbations. The system then smooths due to diffusion and the small random peaks continuously coalesce and react, until a few large domains emerge at $t=0.099$. At $t=0.114$ and $t=0.119$,  multiple patches of higher $a_1$ concentration emerge. The feedback term \eqref{eq_feedback} then plays its role, once the higher $a_2$ concentration location promotes the recruitment of more cytosolic component. This leads to the formation of the single-patch profile at $t=0.159$. From that time until the final time ($t=1$), the spatial configuration only changes in terms of concentration gradients.  File F1 in the ESM contains a movie of the simulation shown in \Cref{fig:tempevol} (A) for both monomeric ($a_1$) and dimeric components ($a_2$).  In \Cref{fig:SFig_5} in the ESM, we show  a similar result for $N=3$.

In order to quantify the single-patch size,  we quantify the surface area of the high-concentration locations in the spherical domain. For this purpose, we define the function
\begin{equation}
 \mathcal{I}^{\varepsilon}_j(t) = \int_{\Gamma} \mathbbm{1}_{\left\{a_j(x,t) > \langle a_j \rangle (t) + \varepsilon \right\}} ds  
 \label{area_fun}
 \end{equation}
 where $\varepsilon$ is the perturbation magnitude, $j$ is the index of the oligomeric component, and $\langle a_j \rangle (t) =   \int_{\Gamma} a_j ds $ is the average concentration of $a_j$ across the sphere $\Gamma$. In Figure  \Cref{fig:SFig_6} 
 in the ESM, we illustrate how the percentage of $\mathcal{S}^\varepsilon_j$ with respect to the total surface area does not change significantly as $\varepsilon$ changes. We then evaluate the evolution of $\mathcal{I}^{\varepsilon}_j(t)$ over time (\Cref{fig:tempevol} (B)).  We exhibit the results of a single simulation for $N=2$ and $N=3$, and  $\varepsilon = 10^{-10}$. At early times, when the concentrations $a_j$ are close to the steady-state $a^*_j$ across the domain, $\mathcal{I}^{\varepsilon}_j(t)$ remains close to 0. Then the combination of diffusion and the feedback term makes the concentration gradients increase in a large portion of the domain, as illustrated in \Cref{fig:tempevol} (A) for $t=0.099$, $t=0.114$, and $t=0.119$ . Finally, the formation of the single-patch promotes the decrease of $\mathcal{I}^{\varepsilon}_j(t)$, since the area of high concentration tends to be small in comparison with the total surface area. Moreover, the concentration outside the patch tends to be small, which makes the average $\langle a_j \rangle (t)$ assume lower values. Therefore, in the final times, the locations in the sphere where the concentrations remain above the average can be associated with the single patch. For this reason, we define the \emph{single-patch area}
 $$\mathcal{S}^\varepsilon_j :=  \mathcal{I}^{\varepsilon}_j(t_f),$$ 
 where $t_f$ is the final simulation time. In this work, we avoid an analytical treatment for the temporal behavior of $\mathcal{I}^{\varepsilon}_j(t_f)$. Instead, we base our definition of the single-patch area on visual inspection of the temporal evolution of $\mathcal{I}^{\varepsilon}_j(t_f)$. In fact, we obtain the same temporal dynamics from \Cref{fig:tempevol} (B) whenever the system forms the single-patch pattern. We also observe that $\mathcal{I}^\varepsilon_1 > \mathcal{I}^\varepsilon_2$ for all times in the case where $N=2$, and also $\mathcal{I}^\varepsilon_1 > \mathcal{I}^\varepsilon_2 > \mathcal{I}^\varepsilon_3 $ in the case where $N=3$. We conclude that $\mathcal{S}^\varepsilon_1 > \mathcal{S}^\varepsilon_2$ (for $N=2$) and $\mathcal{S}^\varepsilon_1 > \mathcal{S}^\varepsilon_2 > \mathcal{S}^\varepsilon_3 $ (for $N=3$). In order to better visualize this area shrinking as the oligomer size increases, we plot the final normalized concentration profiles (\Cref{fig:tempevol} (C)). Given the arc-length parametrization of a geodesic curve crossing the single-patch region, the concentration distributions become tighter for $a_2$ compared to $a_1$ in the case $N=2$. The inset plot shows the non-normalized concentrations, where we see that $a_2>a_1$ in the single-patch location. A similar phenomenon occurs for $N=3$:  the distribution and maximum value of $a_j$ becomes tighter and larger as $j$ increases from 1 to 3.

\begin{figure}[tbp]
\centering
\includegraphics[width=0.85\textwidth]{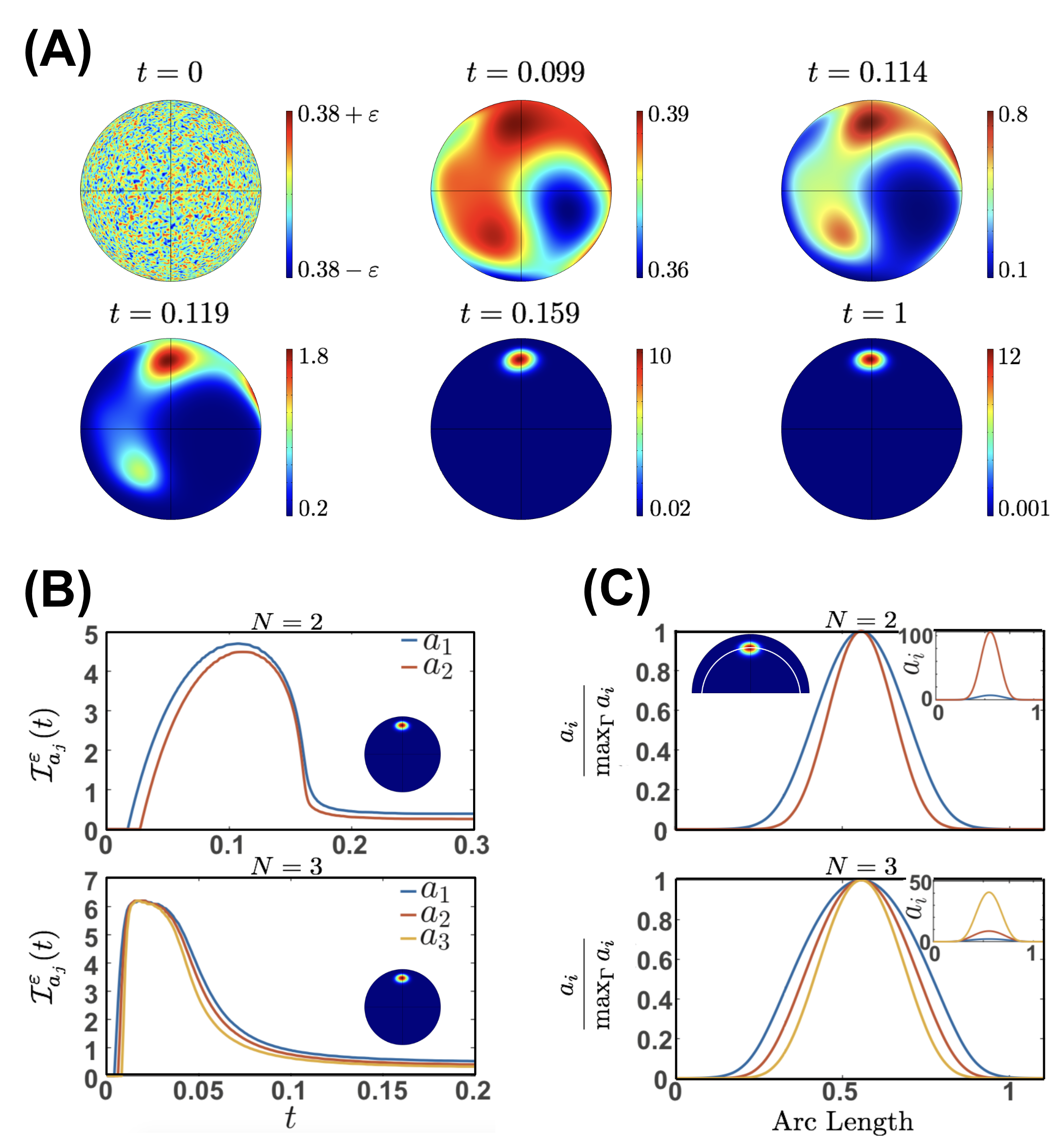}
\caption{\footnotesize\textbf{Temporal Evolution and pattern formation} (A) Spatial distribution of the monomeric component ($a_1$) at different non-dimensional times.  At $t=0$, a random perturbation of magnitude $\varepsilon = 10^{-10}$ is applied to the unstable homogeneous steady-state. At  $t= 0.099$, a small gradient emerges until $t = 0.114$  and $t=0.119$ when the high-concentration  domains begin to coalesce. At  $t=0.159$,  the system converges to the single-patch profile. Finally, at $t= 1$, we show the single-patch steady-state with a final concentration gradient from 0.001 to 12 a.u. In this figure, we consider  $N=2$, $k_0  = 0.025$, and $k_b =2.5$ such that a single steady-state becomes unstable under non-homogeneous perturbations ($(k_0,k_b)$ belongs to Region 1 in \Cref{fig:patternform}). The steady-state is given by  $a^*_1=  0.3817$ , $a^*_2 = 0.1457$, and  $u^* = 0.9806$.  A supplemental movie for panel (A) can be found in supplemental file F1. (B) Evolution of $(\mathcal{I}^{\varepsilon}_{j})(t)$ that gives the single-patch area $\mathcal{S}^{\varepsilon}_{j}$ for $N=2$ and $N=3$  (see text for details). Inset: a single-patch final configuration.  Parameter values: $R=1$, $U=A=13$, $\gamma = 1000$, $d_2=d_3=0.1$, $k_0=0.0161$, $k_m=1$, $k_2=0.4409$. Top:  $N=2$, $k_b=1$. Bottom:  $N=3$, $k_b=10$, $k_g = k_m$, $k_3=k_2$. Initial conditions: $a_1(0)=0.0918$, $a_2(0)=0.0191$, $a_3(0)=0$, $u(0)=2.6099$. (C) For $N=2$ and $N=3$, we plot the final normalized $a_j$ concentrations on a geodesic curve parametrized by arc-length. As the oligomer index $j$ increases, the distribution and maximum value of $a_j$ becomes tighter and higher (inset), respectively.} 
\label{fig:tempevol}
\end{figure}

\subsection{Change of the cell radius and single-patch area}
\label{change_radius}

We investigate how  the single-patch area of a spherical cell depends on its radius $R$. From the non-dimensionalization of the bulk-surface system (see \Cref{nom_dim_subsec}), we defined the characteristic quantities \eqref{char_quan}. In order to move through a dimensional system, we define the dimensionless parameters depending to $R$ to preserve a constant volume concentration. Therefore, we assume a constant $U$ such that $\mathcal{M}_0 \propto R^3$, making the dimensionless parameters as functions of $R$. For this reason, each choice of $R$ will lead to a different solution of the non-dimensional system \eqref{nondim_v2:1}--\eqref{nondim_v2:6}. In particular, it will also change the non-dimensional single-patch area $\mathcal{S}^\varepsilon_j$. We show the results for $R$ ranging from $0.5$ to $5$ and two different parameters (\Cref{fig:changeradius}): 
$$ \text{area percentage} =  \frac{\mathcal{S}^\varepsilon_j}{4 \pi} \times 100 \quad \text{and} \quad \text{dimensional area} = \mathcal{S}^\varepsilon_j \ R^2.$$

For $N=2$ and $N=3$, we provide the same total volume concentration for the system. For clarity, in this section we will refer to the non-dimensional system with the hat ($\hat{\textcolor{white}{a}}$) notation. We define the initial conditions as a linear ramp of slope $\varepsilon$ around the steady-state
$$\hat{a}_1(x,0) = \hat{a}^*_1 + \varepsilon  \ \hat{x}_1 \quad \text{and} \quad  \hat{a}_2(x,0) = \hat{a}^*_2 - \frac{1}{2}  \  \varepsilon \ \hat{x}_1, $$
 where $\hat{x}=(\hat{x}_1,\hat{x}_2,\hat{x}_3) \in \hat{\Gamma}$ and $\hat{x}_1$ is the position with the sphere centered in the origin. In the case $N=3$, we assumed $a_3^*=0$. See \Cref{fig:changeradius} caption for details on the parameter choices. We then observe the same phenomena: the dimensional area (\cref{fig:changeradius} red; open circles) increases approximately linearly with $R$. On the other hand, the area percentage (\cref{fig:changeradius} black; closed circles) decreases with $R$. We can then conclude that although the dimensional area of clusters increases, the additional spherical area changes at a much faster rate since the area percentage varies with $\approx \frac{1}{R}$. The effect of $N$ on single-patch area shows both dimensional area and area percentage  are higher for $N=3$ (\cref{fig:changeradius}(B)) in comparison with $N=2$  (\cref{fig:changeradius}(A)).

\begin{figure}[tbp]
\centering
\includegraphics[width=0.90\textwidth]{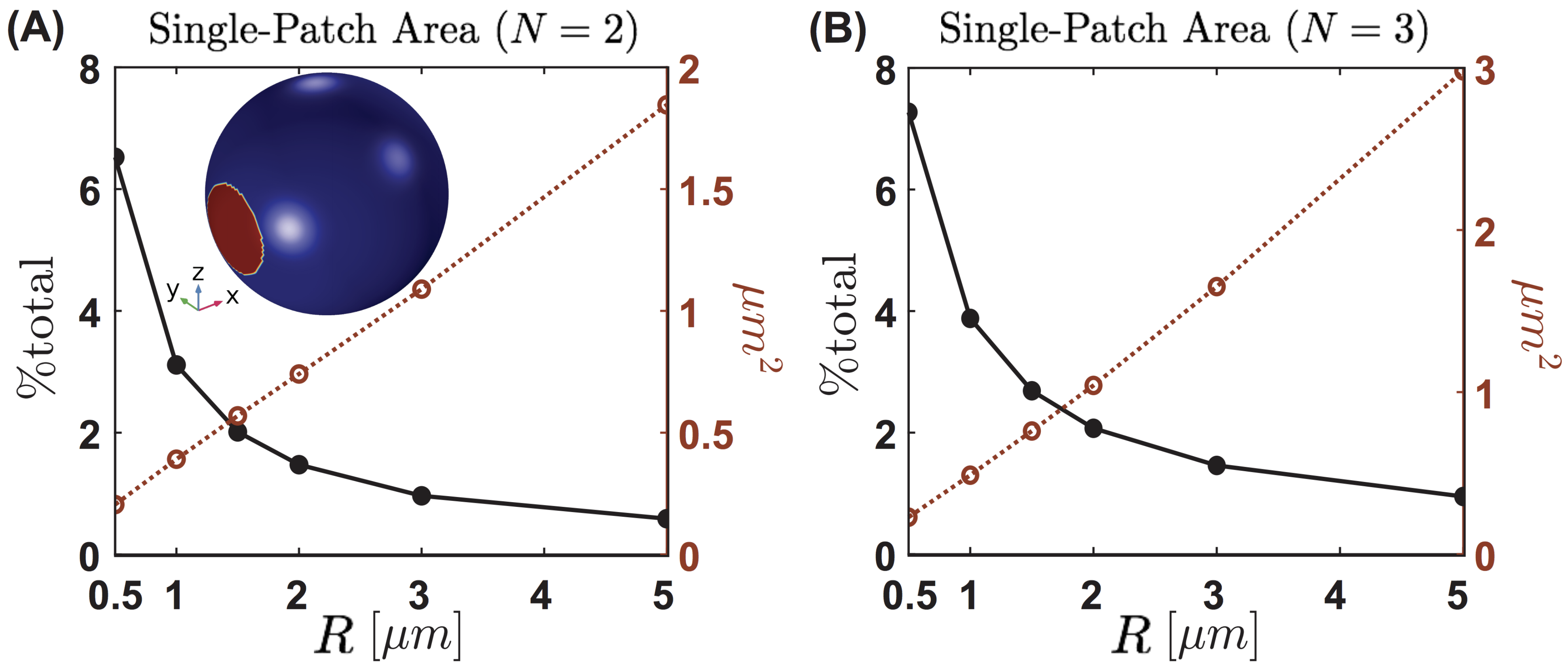}
\caption{\footnotesize\textbf{Change of the cell radius and single-patch area.} We quantify the percentage of the total area and the dimensional area (see text for details), for various radius $R$ ranging from $0.5$ to $5$. The $R$ value was changed in the non-dimensional system with a fixed concentration ($U$) through variations in $\Gamma$, $\gamma$, $A$, $\hat{k}_0$, $\hat{k}_m$, and $\hat{k}_g$. (A) We quantify the Patch size for the $N=2$ case (red; open circles), then normalized against the total area of the sphere (black; closed circles). As the radius increases, the patch size increases approximately linearly, but the percent area decreases rapidly. (B) The same simulation for $N=3$. As the radius increases the patch size increases, but the total percent area still decreases. Between cases,  we observe the same general qualitative properties for single-patch area percentage and dimensional area. The major differences arise in the absolute values, as $N=3$ creates larger patches.}
\label{fig:changeradius}
\end{figure}


\section{Discussion}
 \label{sec:diss}

Protein heterogeneity in the PM is of critical importance to cellular functions. 
Many factors influence this heterogeneity, including membrane composition, protein-protein interaction, phase separation, lateral diffusion, and possible feedback, resulting in the formation of spatial patterns  \cite{hashimoto2010,johannes2018,ispolatov2005}.
For this reason, understanding the interplay of aggregation kinetics, lateral diffusion, and feedback in the formation of spatial patterns is an essential step towards developing a complete description of the mechanisms behind protein clustering on the cell surface. 
In this work, we developed a bulk-surface model for protein aggregation with positive feedback that exhibits a spatial heterogeneous single-patch steady-state.  
To the best of our knowledge, this is the first modeling attempt that merges the reaction-diffusion version of classical Smoluchowski dynamics with the modern bulk-surface geometrical setup. 
 
A major result from our model is the role played by the feedback term $k_b \ a_N$ in the boundary conditions \eqref{eq_feedback}. If $k_b$ is low enough, the steady-state distribution is spatially uniform, and no protein heterogeneity exists. 
For $N=2$, we formally proved such a result (\Cref{thm1}), and for $N=3$, we used numerical simulations to observe a similar phenomenon. 
In particular, in the total absence of feedback ($k_b=0$), we observed that spatial heterogeneity is not achievable when we only considered protein-protein interaction. On the other hand, if $k_b$ is sufficiently high, we observed the emergence of linear instability and therefore patterning on the cellular surface. Experimental observations have shown that membrane proteins do organize in a spatially heterogeneous fashion \cite{choquet2010,gan2015,padmanabhan2019}.
However, the molecular mechanisms are still being investigated experimentally. 
 The feedback mechanism we proposed here can also be interpreted in purely biological terms. The largest oligomers recruit ligands from the cytosol, which form ligand-receptor monomers.  
If the rate of recruitment of monomers is low, diffusive effects dominate, and the configuration of the system is homogeneous in space. 
On the other hand, a higher rate promotes a significant input of new monomeric components. 
Then, continued oligomerization generates higher concentrations of the largest components, which closes the  positive feedback loop and drives  pattern formation. 
The largest oligomers can thus be interpreted as self-activators of pattern formation.
For this reason, our mechanism of pattern formation can be related to the classical Turing framework where self-activation is required to generate spatial patterns \cite{turing1952, Gierer1972}.  Another interesting aspect of our model is the absence of an explicit description of cooperative binding.  For the wave-pinning model \cite{mori2008a, mori2011,Cusseddu2018}, cooperativity is included with a Hill function, which accounts for the positive feedback. In contrast, our oligomerization reactions assume only mass-action kinetics, which seems to be insufficient for pattern formation without the feedback term.

Bistable systems are well known to promote diffusion-driven instabilities in the context of cell polarization \cite{rappel2017mechanisms,semplice2012bistable}.  For the wave-pinning model \cite{mori2008a,mori2011,Cusseddu2018}, the structure of the Hill function is responsible for  bistability. 
Other studies followed a similar approach,  using a particular choice of reaction flux that is naturally associated with a bistable regime \cite{Beta2008,Alonso2010}.
In our model, bistability emerges by the combination of two key ingredients: positive feedback and mass conservation.
This observation becomes clear as we carefully inspect the steady-state analysis of the reduced system (cf. \Cref{hom_ss_subsec}).
First, the equilibrium of the oligomerization reactions (driven by mass action kinetics only) provides the distribution across the different surface components. 
Then, the input from the non-local functional comes into play, as a consequence of the boundary conditions and mass conservation. 
The non-local functional at steady-state provides an extra equation, which gives the equilibrium solutions for the monomeric component.
The particular contribution of the feedback comes from the coefficient  $C_N \left( k_0 |\Gamma| N  - \mathcal{M}_0 k_b  \right) \alpha^N$ of the polynomial $\mathcal{P}_N(\alpha)$.
If the coefficient is negative, then the existence of three roots, and therefore three steady-states, is achievable.
In this case, we can compute their stability under homogeneous perturbations and verify bistability.

 Under non-homogeneous perturbations, one or two stable steady-states may become unstable, and the system undergoes a diffusion-driven instability.  
 Even more impressive is the emergence of a linear instability parameter region, called Region 1 in this study, when the system admits a single steady-state that becomes unstable.  
 We note that in Getz et al.~\cite{getz2018stability}, the authors were able to find a region of linear instability for the Wave-Pinning model that is comparable with our Region 1. 
 While the authors briefly discussed the changes in that parameter region for different wave-numbers, here we explicitly showed that the leading eigenmode exhibits a region of instability that shrinks as the eigenmode index increases.
 Such instability in the lower modes, which are associated to the smallest positive eigenvalues of the Laplace-Beltrami operator,  has been ofter related to a single-patch steady-state pattern \cite{Ratz2014,goryachev2008dynamics}, which is confirmed for our system.

The single-patch steady-state consistently appears for parameter values corresponding to the different instability regions (called as Regions 0, 1,2 and 3).
Goryachev et al.  ~\cite{goryachev2008dynamics} found a similar spatial profile for the  Cdc42 GTPase cycle, where the income of new cytoplasmic components maintained the cluster steady-state and compensated for its lateral diffusion. 
A similar phenomenon seems to happen in our system.  
An allegory that explains the stable existence of such heterogeneous steady-states is the so-called ``rich get richer" competition \cite{manor2006dynamical}. 
In this case,  larger domains outcompete for the smaller until only one stable domain arises. 
Our hypothesis about the existence of the single-patch is based on the role of the positive feedback term. 
We assume that the presence of high concentrations of the largest oligomer promotes ligand binding onto the PM in a linear fashion, without any saturation mechanism nor steric effects. 
As in \cite{goryachev2008dynamics}, this assumption seems to account for a resource competition that excludes the possibility of multiple patches.

Based on the insights from our model, we identify future research directions that will enhance studies such as ours. 
In the current formulation, we lack a formal explanation for the emergence and robustness of the single patch steady-state. 
The spatial aspects of the model render such analysis hard, but it may be possible to obtain a formal proof by considering a one-dimensional version of our system as in \cite{mori2008a}. Another interesting quantity to be computed in future studies is the so-called \emph{amplitude of the pattern}, for which a formal calculation was recently developed \cite{chen2019}.
Additionally, the mathematical challenge for a theoretical stability result lies in the increasing complexity of the system as $N$ increases. 
In this case, we have relied on numerical simulations for $N=3$ to identify the threshold phenomenon for diffusion-driven instabilities. 
However, future efforts in this direction could open up new mathematical avenues for stability analysis of increasingly complex systems.
Finally, including the role of curvature and cytosolic diffusion in the formation of membrane protein aggregates would bring us closer to analyses of biological and biophysical systems. These are the focus on ongoing studies in our group.

\section{Acknowledgments} 

This work was supported by Air Force Office of Scientific Research (AFOSR)
Multidisciplinary University Research Initiative (MURI) grant
FA9550-18-1-0051 to P. Rangamani. M. Holst was supported in part by NSF Awards DMS 1620366 and DMS 1345013.

\footnotesize{
\bibliography{IB_ref} 

\begin{thebibliography}{10}

\bibitem{darnell1990molecular}
J.~E. Darnell, H.~F. Lodish, D.~Baltimore, et~al.
\newblock {\em Molecular cell biology}, volume~2.
\newblock Scientific American Books New York, 1990.

\bibitem{stillwell2013introduction}
W.~Stillwell.
\newblock {\em An introduction to biological membranes: from bilayers to
  rafts}.
\newblock Newnes, 2013.

\bibitem{guidotti1972composition}
G.~Guidotti.
\newblock The composition of biological membranes.
\newblock {\em Archives of internal medicine}, 129(2):194--201, 1972.

\bibitem{yeagle2011structure}
P.~L. Yeagle.
\newblock {\em The structure of biological membranes}.
\newblock CRC press, 2011.

\bibitem{albersheim1975carbohydrates}
P.~Albersheim and A.~J. Anderson-Prouty.
\newblock Carbohydrates, proteins, cell surfaces, and the biochemistry of
  pathogenesis.
\newblock {\em Annual Review of Plant Physiology}, 26(1):31--52, 1975.

\bibitem{jain1988introduction}
M.~K. Jain, R.~C. Wagner, et~al.
\newblock Introduction to biological membranes.
\newblock 1988.

\bibitem{hashimoto2010}
K.~Hashimoto and A.~R. Panchenko.
\newblock Mechanisms of protein oligomerization, the critical role of
  insertions and deletions in maintaining different oligomeric states.
\newblock {\em Proceedings of the National Academy of Sciences},
  107(47):20352--20357, 2010.

\bibitem{johannes2018}
L.~Johannes, W.~Pezeshkian, J.~H. Ipsen, and J.~C. Shillcock.
\newblock Clustering on {Membranes}: {Fluctuations} and {More}.
\newblock {\em Trends in Cell Biology}, 28(5):405--415, 2018.

\bibitem{ispolatov2005}
I.~Ispolatov.
\newblock Binding properties and evolution of homodimers in protein-protein
  interaction networks.
\newblock {\em Nucleic Acids Research}, 33(11):3629--3635, 2005.

\bibitem{porat-shliom2013}
N.~Porat-Shliom, O.~Milberg, A.~Masedunskas, and R.~Weigert.
\newblock Multiple roles for the actin cytoskeleton during regulated
  exocytosis.
\newblock {\em Cellular and molecular life sciences : CMLS}, 70(12):2099--2121,
  2013.

\bibitem{lorent2017}
J.~H. Lorent, B.~Diaz-Rohrer, X.~Lin, K.~Spring, A.~A. Gorfe, K.~R. Levental,
  and I.~Levental.
\newblock Structural determinants and functional consequences of protein
  affinity for membrane rafts.
\newblock {\em Nature Communications}, 8(1):1219, 2017.

\bibitem{Mori2008}
Y.~Mori, A.~Jilkine, and L.~Edelstein-Keshet.
\newblock Wave-pinning and cell polarity from a bistable reaction-diffusion
  system.
\newblock {\em Biophysical journal}, 94(9):3684--3697, 2008.

\bibitem{lao2010}
Q.~Z. Lao, E.~Kobrinsky, Z.~Liu, and N.~M. Soldatov.
\newblock Oligomerization of {Ca$_v\beta$} subunits is an essential correlate
  of {Ca}2+ channel activity.
\newblock {\em FASEB journal : official publication of the Federation of
  American Societies for Experimental Biology}, 24(12):5013--5023, 2010.

\bibitem{lemmon2010}
M.~A. Lemmon and J.~Schlessinger.
\newblock Cell signaling by receptor tyrosine kinases.
\newblock {\em Cell}, 141(7), 2010.

\bibitem{sleno2018}
R.~Sleno and T.~E. HÃ©bert.
\newblock Chapter {Five} - {The} {Dynamics} of {GPCR} {Oligomerization} and
  {Their} {Functional} {Consequences}.
\newblock In A.~K. Shukla, editor, {\em International {Review} of {Cell} and
  {Molecular} {Biology}}, volume 338 of {\em G {Protein}-{Coupled} {Receptors}:
  {Emerging} {Paradigms} in {Activation}, {Signaling} and {Regulation} {Part}
  {A}}, pages 141--171. Academic Press, 2018.

\bibitem{baisamy2005}
L.~Baisamy, N.~Jurisch, and D.~Diviani.
\newblock Leucine {Zipper}-mediated {Homo}-oligomerization {Regulates} the
  {Rho}-{GEF} {Activity} of {AKAP}-{Lbc}.
\newblock {\em Journal of Biological Chemistry}, 280(15):15405--15412, 2005.

\bibitem{chen2005}
C.~P. Chen, S.~Posy, A.~Ben-Shaul, L.~Shapiro, and B.~H. Honig.
\newblock Specificity of cell-cell adhesion by classical cadherins: {Critical}
  role for low-affinity dimerization through -strand swapping.
\newblock {\em Proceedings of the National Academy of Sciences},
  102(24):8531--8536, 2005.

\bibitem{askarova2011}
S.~Askarova, X.~Yang, and J.~C.-M. Lee.
\newblock Impacts of {Membrane} {Biophysics} in {Alzheimer}'s {Disease}: {From}
  {Amyloid} {Precursor} {Protein} {Processing} to {$\beta$} {Peptide}-{Induced}
  {Membrane} {Changes}.
\newblock {\em International Journal of Alzheimer's Disease}, 2011, 2011.

\bibitem{sarkar2013}
B.~Sarkar, A.~Das, and S.~Maiti.
\newblock Thermodynamically stable amyloid-$\beta$ monomers have much lower
  membrane affinity than the small oligomers.
\newblock {\em Frontiers in Physiology}, 4:84, 2013.

\bibitem{zhang2012}
Y.-J. Zhang, J.-M. Shi, C.-J. Bai, H.~Wang, H.-Y. Li, Y.~Wu, and S.-R. Ji.
\newblock Intra-membrane {Oligomerization} and {Extra}-membrane
  {Oligomerization} of {Amyloid}-$\beta$ {Peptide} {Are} {Competing}
  {Processes} as a {Result} of {Distinct} {Patterns} of {Motif} {Interplay}.
\newblock {\em Journal of Biological Chemistry}, 287(1):748--756, 2012.

\bibitem{andreasen2015}
M.~Andreasen, N.~Lorenzen, and D.~Otzen.
\newblock Interactions between misfolded protein oligomers and membranes: {A}
  central topic in neurodegenerative diseases?
\newblock {\em Biochimica et Biophysica Acta (BBA) - Biomembranes},
  1848(9):1897--1907, 2015.

\bibitem{Habchi2018}
J.~Habchi, S.~Chia, C.~Galvagnion, T.~C.~T. Michaels, M.~M.~J. Bellaiche, F.~S.
  Ruggeri, M.~Sanguanini, I.~Idini, J.~R. Kumita, E.~Sparr, S.~Linse, C.~M.
  Dobson, T.~P.~J. Knowles, and M.~Vendruscolo.
\newblock Cholesterol catalyses {$\beta$}42 aggregation through a heterogeneous
  nucleation pathway in the presence of lipid membranes.
\newblock {\em Nature Chemistry}, 10(6):673--683, 2018.

\bibitem{choquet2010}
D.~Choquet.
\newblock Fast {AMPAR} trafficking for a high-frequency synaptic transmission.
\newblock {\em European Journal of Neuroscience}, 32(2):250--260, 2010.

\bibitem{gan2015}
Q.~Gan, C.~L. Salussolia, and L.~P. Wollmuth.
\newblock Assembly of {AMPA} receptors: mechanisms and regulation.
\newblock {\em The Journal of Physiology}, 593(Pt 1):39--48, 2015.

\bibitem{padmanabhan2019}
P.~Padmanabhan, R.~Martinez-Mairmol, D.~Xia, J.~Gotz, and F.~A. Meunier.
\newblock Frontotemporal dementia mutant {Tau} promotes aberrant {Fyn}
  nanoclustering in hippocampal dendritic spines.
\newblock {\em eLife}, 8:e45040, 2019.

\bibitem{Rangamani2013}
P.~Rangamani, A.~Lipshtat, E.~U. Azeloglu, R.~C. Calizo, M.~Hu, S.~Ghassemi,
  J.~Hone, S.~Scarlata, S.~R. Neves, and R.~Iyengar.
\newblock Decoding information in cell shape.
\newblock {\em Cell}, 154(6):1356--1369, 2013.

\bibitem{Frey2018}
E.~Frey, J.~Halatek, S.~Kretschmer, and P.~Schwille.
\newblock Protein pattern formation.
\newblock In {\em Physics of Biological Membranes}, pages 229--260. Springer,
  2018.

\bibitem{Denk2018}
J.~Denk, S.~Kretschmer, J.~Halatek, C.~Hartl, P.~Schwille, and E.~Frey.
\newblock Mine conformational switching confers robustness on self-organized
  min protein patterns.
\newblock {\em Proceedings of the National Academy of Sciences},
  115(18):4553--4558, 2018.

\bibitem{Cusseddu2018}
D.~Cusseddu, L.~Edelstein-Keshet, J.~A. Mackenzie, S.~Portet, and
  A.~Madzvamuse.
\newblock A coupled bulk-surface model for cell polarisation.
\newblock {\em Journal of theoretical biology}, 2018.

\bibitem{Giese2015}
W.~Giese, M.~Eigel, S.~Westerheide, C.~Engwer, and E.~Klipp.
\newblock Influence of cell shape, inhomogeneities and diffusion barriers in
  cell polarization models.
\newblock {\em Physical biology}, 12(6):066014, 2015.

\bibitem{Diegmiller2018}
R.~Diegmiller, H.~Montanelli, C.~B. Muratov, and S.~Y. Shvartsman.
\newblock Spherical caps in cell polarization.
\newblock {\em Biophysical journal}, 115(1):26--30, 2018.

\bibitem{Madzvamuse2015}
A.~Madzvamuse, A.~H.~W. Chung, and C.~Venkataraman.
\newblock Stability analysis and simulations of coupled bulk-surface
  reaction-diffusion systems.
\newblock {\em Proceedings of the Royal Society A: Mathematical, Physical and
  Engineering Sciences}, 471(2175):20140546--20140546, 2015.

\bibitem{Ratz2012}
A.~R{\"a}tz and M.~R{\"o}ger.
\newblock Turing instabilities in a mathematical model for signaling networks.
\newblock {\em Journal of mathematical biology}, 65(6-7):1215--1244, 2012.

\bibitem{Ratz2015}
A.~R{\"a}tz.
\newblock Turing-type instabilities in bulk--surface reaction--diffusion
  systems.
\newblock {\em Journal of Computational and Applied Mathematics}, 289:142--152,
  2015.

\bibitem{Smoluchowski1918}
M.~v. Smoluchowski.
\newblock Versuch einer mathematischen theorie der koagulationskinetik
  kolloider l{\"o}sungen.
\newblock {\em Zeitschrift f{\"u}r physikalische Chemie}, 92(1):129--168, 1918.

\bibitem{Drake1972}
R.~L. Drake.
\newblock A general mathematical survey of the coagulation equation.
\newblock {\em Topics in current aerosol research (Part 2)}, 3(Part
  2):201--376, 1972.

\bibitem{Arosio2012}
P.~Arosio, S.~Rima, M.~Lattuada, and M.~Morbidelli.
\newblock Population balance modeling of antibodies aggregation kinetics.
\newblock {\em The Journal of Physical Chemistry B}, 116(24):7066--7075, 2012.

\bibitem{Zidar2018}
M.~Zidar, D.~Kuzman, and M.~Ravnik.
\newblock Characterisation of protein aggregation with the smoluchowski
  coagulation approach for use in biopharmaceuticals.
\newblock {\em Soft matter}, 14(29):6001--6012, 2018.

\bibitem{Achdou2013}
Y.~Achdou, B.~Franchi, N.~Marcello, and M.~C. Tesi.
\newblock A qualitative model for aggregation and diffusion of $\beta$-amyloid
  in alzheimer’s disease.
\newblock {\em Journal of mathematical biology}, 67(6-7):1369--1392, 2013.

\bibitem{Franchi2016}
B.~Franchi and S.~Lorenzani.
\newblock From a microscopic to a macroscopic model for alzheimer disease:
  two-scale homogenization of the smoluchowski equation in perforated domains.
\newblock {\em Journal of Nonlinear Science}, 26(3):717--753, 2016.

\bibitem{Bertsch2016}
M.~Bertsch, B.~Franchi, N.~Marcello, M.~C. Tesi, and A.~Tosin.
\newblock Alzheimer's disease: a mathematical model for onset and progression.
\newblock {\em Mathematical medicine and biology: a journal of the IMA},
  34(2):193--214, 2016.

\bibitem{Bentz1981}
J.~Bentz and S.~Nir.
\newblock Mass action kinetics and equilibria of reversible aggregation.
\newblock {\em Journal of the Chemical Society, Faraday Transactions 1:
  Physical Chemistry in Condensed Phases}, 77(6):1249--1275, 1981.

\bibitem{Changeux1967}
J.-P. Changeux, J.~Thi{\'e}ry, Y.~Tung, and C.~Kittel.
\newblock On the cooperativity of biological membranes.
\newblock {\em Proceedings of the National Academy of Sciences of the United
  States of America}, 57(2):335, 1967.

\bibitem{Oosterom2006}
A.~van Oosterom.
\newblock The surface laplacian operator of the potentials on a bounded volume
  conductor has a unique inverse.
\newblock {\em IEEE transactions on biomedical engineering}, 53(7):1449--1450,
  2006.

\bibitem{postma2004}
M.~Postma, L.~Bosgraaf, H.~M. Loovers, and P.~J. Van~Haastert.
\newblock Chemotaxis: signalling modules join hands at front and tail.
\newblock {\em EMBO reports}, 5(1):35--40, 2004.

\bibitem{turing1952}
A.~M. Turing.
\newblock The {{Chemical Basis}} of {{Morphogenesis}}.
\newblock {\em Philosophical Transactions of the Royal Society of London.
  Series B, Biological Sciences}, 237(641):37--72, 1952.

\bibitem{strogatz1994}
S.~H. Strogatz.
\newblock {\em Nonlinear {{Dynamics And Chaos}}: {{With Applications To
  Physics}}, {{Biology}}, {{Chemistry And Engineering}}}.
\newblock {Westview Press}, first edition edition edition, 1994.

\bibitem{Sharma2016}
V.~Sharma and J.~Morgan.
\newblock Global existence of solutions to reaction-diffusion systems with mass
  transport type boundary conditions.
\newblock {\em SIAM Journal on Mathematical Analysis}, 48(6):4202--4240, 2016.

\bibitem{Jero83}
J.~Jerome.
\newblock {\em Approximation of Nonlinear Evolution Systems}.
\newblock Academic Press, New York, NY, 1983.

\bibitem{Smith2018}
S.~Smith and N.~Dalchau.
\newblock Model reduction permits turing instability analysis of arbitrary
  reaction-diffusion models.
\newblock {\em Journal of the Royal Society Interface}, 15, 2018.

\bibitem{murray1993}
J.~D. Murray.
\newblock {\em Mathematical {{Biology}}}.
\newblock {Springer}, 2nd corr edition, 1993.

\bibitem{Gierer1972}
A.~Gierer and H.~Meinhardt.
\newblock A theory of biological pattern formation.
\newblock {\em Kybernetik}, 12(1):30--39, 1972.

\bibitem{mori2008a}
Y.~Mori, A.~Jilkine, and L.~Edelstein-Keshet.
\newblock Wave-pinning and cell polarity from a bistable reaction-diffusion
  system.
\newblock {\em Biophysical Journal}, 94(9):3684--3697, 2008.

\bibitem{mori2011}
Y.~Mori, A.~Jilkine, and L.~Edelstein-Keshet.
\newblock Asymptotic and {{Bifurcation Analysis}} of {{Wave}}-{{Pinning}} in a
  {{Reaction}}-{{Diffusion Model}} for {{Cell Polarization}}.
\newblock {\em SIAM Journal on Applied Mathematics}, 71(4):1401--1427, 2011.

\bibitem{rappel2017mechanisms}
W.-J. Rappel and L.~Edelstein-Keshet.
\newblock Mechanisms of cell polarization.
\newblock {\em Current opinion in systems biology}, 3:43--53, 2017.

\bibitem{semplice2012bistable}
M.~Semplice, A.~Veglio, G.~Naldi, G.~Serini, and A.~Gamba.
\newblock A bistable model of cell polarity.
\newblock {\em PloS one}, 7(2):e30977, 2012.

\bibitem{Beta2008}
C.~Beta, G.~Amselem, and E.~Bodenschatz.
\newblock A bistable mechanism for directional sensing.
\newblock {\em New Journal of Physics}, 10(8):083015, 2008.

\bibitem{Alonso2010}
S.~Alonso and M.~Baer.
\newblock Phase separation and bistability in a three-dimensional model for
  protein domain formation at biomembranes.
\newblock {\em Physical biology}, 7(4):046012, 2010.

\bibitem{getz2018stability}
M.~C. Getz, J.~A. Nirody, and P.~Rangamani.
\newblock Stability analysis in spatial modeling of cell signaling.
\newblock {\em Wiley Interdisciplinary Reviews: Systems Biology and Medicine},
  10(1):e1395, 2018.

\bibitem{Ratz2014}
A.~R{\"a}tz and M.~R{\"o}ger.
\newblock Symmetry breaking in a bulk--surface reaction--diffusion model for
  signalling networks.
\newblock {\em Nonlinearity}, 27(8):1805, 2014.

\bibitem{goryachev2008dynamics}
A.~B. Goryachev and A.~V. Pokhilko.
\newblock Dynamics of cdc42 network embodies a turing-type mechanism of yeast
  cell polarity.
\newblock {\em FEBS letters}, 582(10):1437--1443, 2008.

\bibitem{manor2006dynamical}
A.~Manor and N.~M. Shnerb.
\newblock Dynamical failure of turing patterns.
\newblock {\em EPL (Europhysics Letters)}, 74(5):837, 2006.

\bibitem{chen2019}
Y.~Chen and J.~Buceta.
\newblock A non-linear analysis of turing pattern formation.
\newblock {\em PLOS ONE}, 14(8):1--9, 2019.

\end{thebibliography}
\bibliographystyle{unsrtPR.bst}
}

\clearpage


\section{Electronic Supplementary Material}
\beginsupplement

\begin{figure}[h]
\centering
\includegraphics[width=0.85\textwidth]{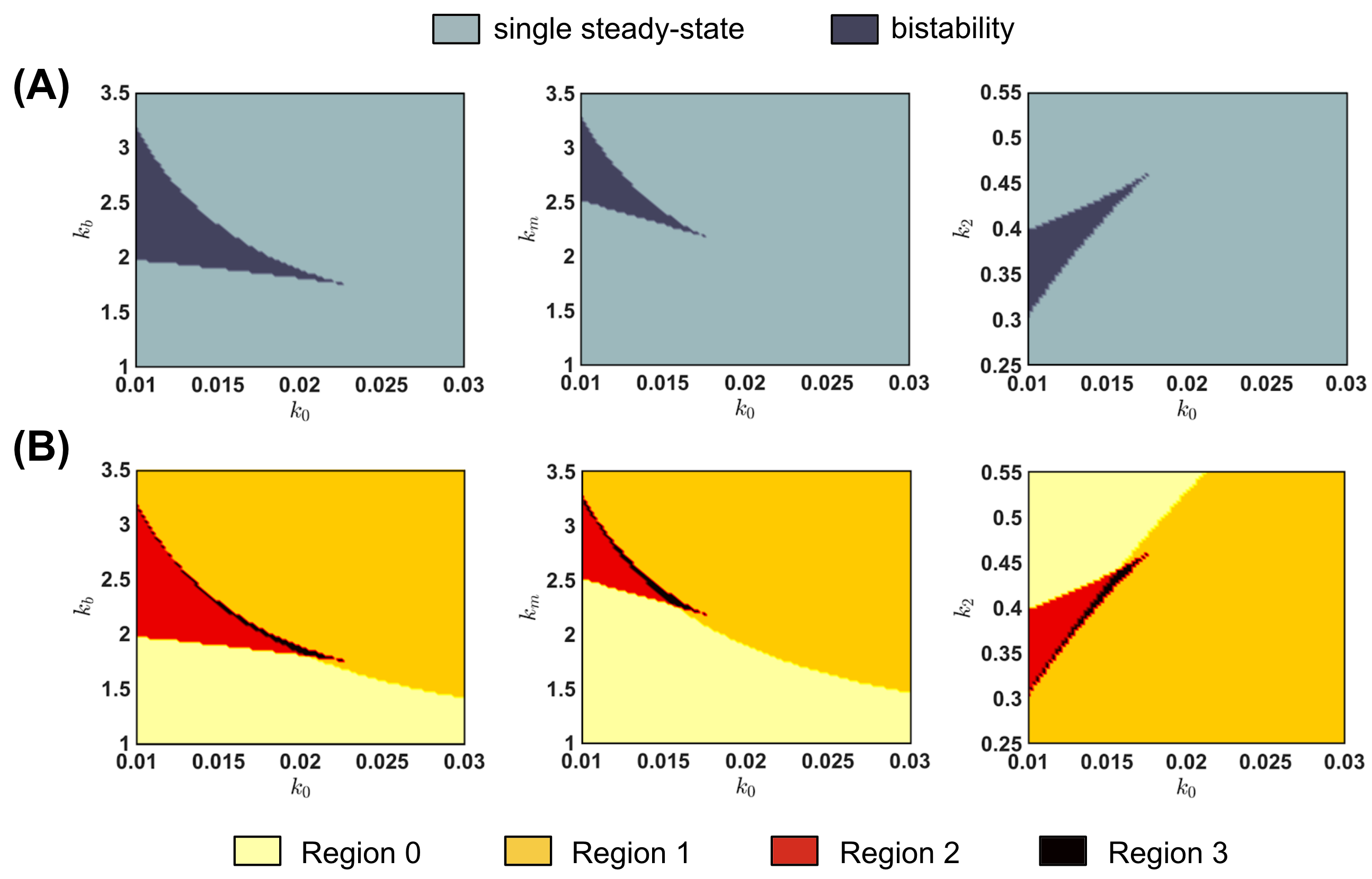}
\caption{\footnotesize\textbf{Parameter Regions of Bistability and Linear Instability ($N=2$).}     We scan the reaction rates for different parameter values. (A) regions where the well-mixed system exhibits bistability. (B) The correspondent Regions 0, 1, 2, and 3 (see manuscript for details).   In the figure, we fixed   $d_2=0.1$,  $\gamma = 1000$, and eigenmode index $l = 1$. 
}
\label{fig:SFig_1}
\end{figure}

\begin{figure}[h]
\centering
\includegraphics[width=0.85\textwidth]{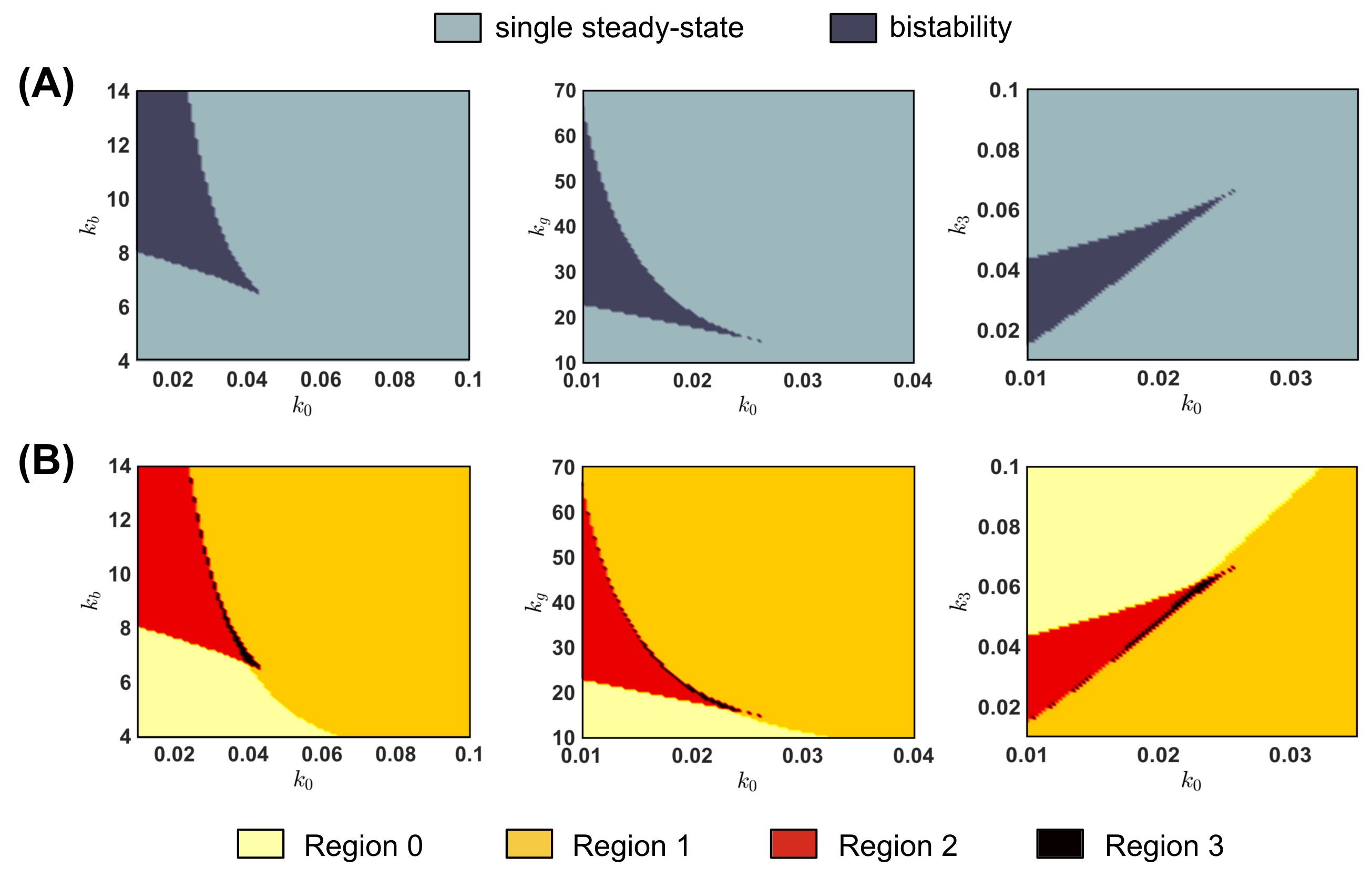}
\caption{\footnotesize\textbf{Parameter Regions of Bistability and Linear Instability ($N=3$).}     We scan the reaction rates for different parameter values. (A) regions where the well-mixed system exhibits bistability. (B) The correspondent Regions 0, 1, 2, and 3 (see manuscript for details).  In the figure, we fixed   $d_2=d_3=0.1$,  $\gamma= 1000$, and eigenmode index $l = 1$.}
\label{fig:SFig_2}
\end{figure}

\clearpage

\begin{figure}[h]
\centering
\includegraphics[width=0.85\textwidth]{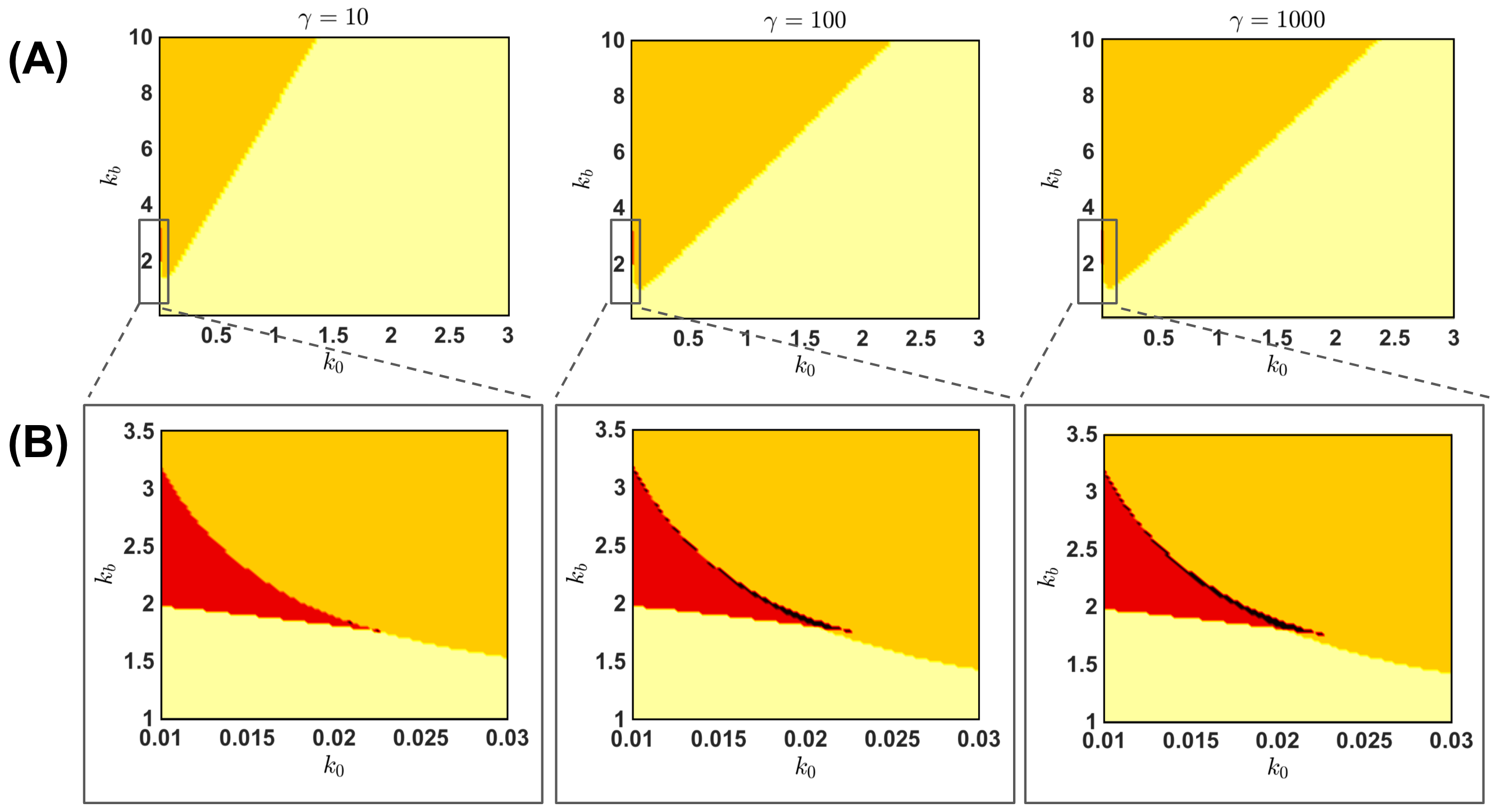}
\caption{\footnotesize\textbf{Changing the dimensionless parameter $\gamma$  ($N=2$).}  (A) Changing the reaction parameter $\gamma$ for a wider range of $k_0$ and $k_b$ allow us to observe instability regions in the single steady-state regime that are considerably larger than  the union of Regions 2 and 3. We observe an increase of Region 1 as $\gamma$ increases. (B) A zoom on Region 2 and 3 shows little differences among the profiles, except for $\gamma=10$, where the Region 3 is significantly reduced. In this figure, we consider $d_2=0.1$ and eigenmode $l=1$.}
\label{fig:SFig_3}
\end{figure}

\vspace*{-0.25cm}

\begin{figure}[h]
\centering
\includegraphics[width=0.80\textwidth]{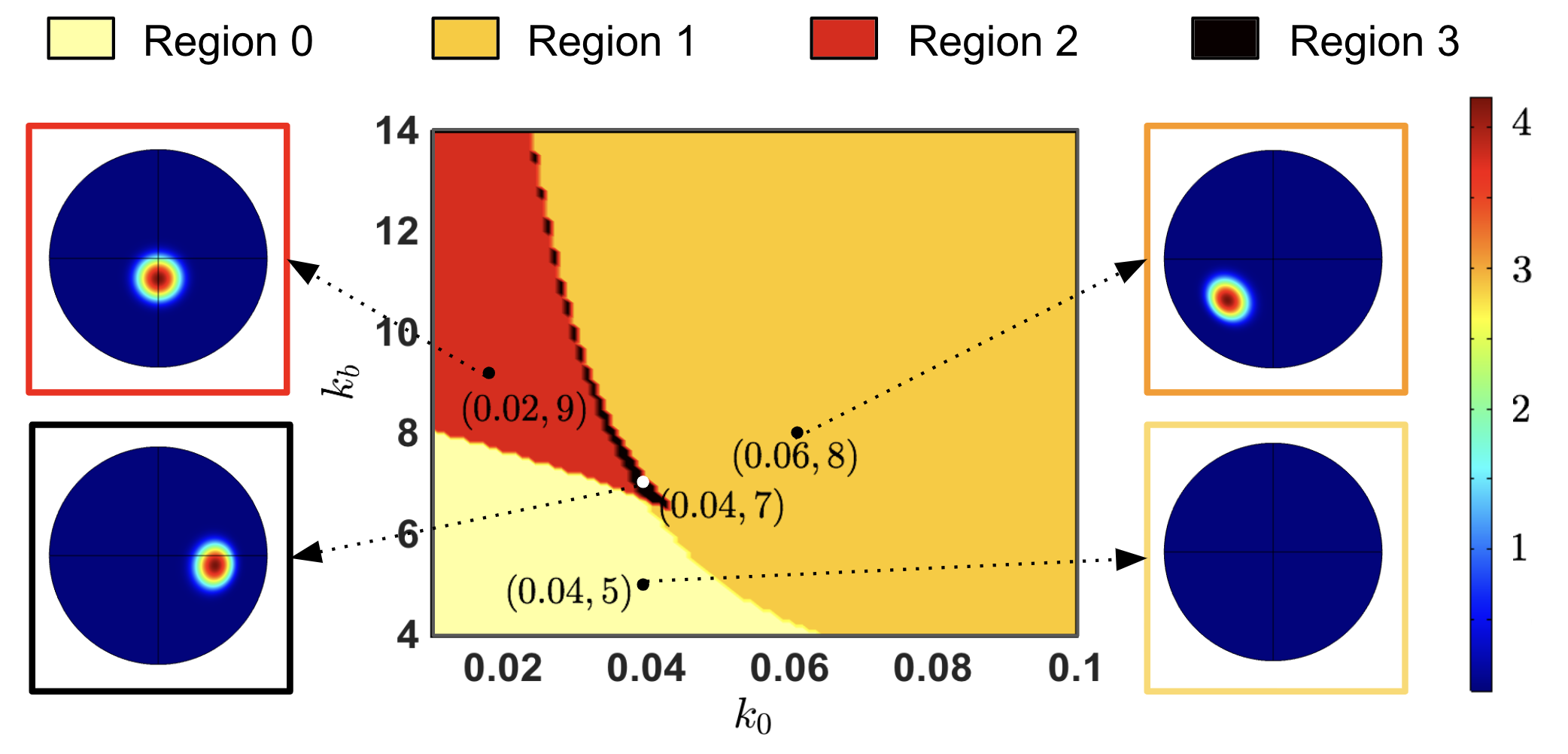}
\caption{\footnotesize\textbf{Linear Instability and Pattern Formation ($N=3$)}.  We exhibit the stability analysis colormap  for eigenmode index $l=1$ and the final spatial profile of the $a_1$ component. We consider four $(k_0,k_b)$ values from Regions 0, 1, 2, and 3, which are colored in light-yellow, orange, red or black, respectively. For Regions 1, 2, and 3, we observe the emergence of a single-patch spatially heterogeneous steady-state, which is consistent across parameter regions in terms of its circular shape and concentration gradient. For Region 0, we do not observe a pattern formation for this particular eigenmode.  In the figure, we fixed   $d_2=d_3=0.1$,  $\gamma= 1000$,  eigenmode index $l = 1$, $k_m = k_g= k_2=k_3 =1$. steady-state values.  Region 0: $a^*_1=0.1251$, $a^*_2=0.0157$, $a^*_3=0.002$, $u^* =2.513$, Region 1: $a^*_1=0.3439$, $a^*_2=0.1183$, $a^*_3=0.0407$, $u^* =0.8922$. Region 2: $a^*_1=0.3442$, $a^*_2=0.1185$, $a^*_3=0.0408$, $u^* =0.8892$. Region 3: $a^*_1=0.1598$, $a^*_2=0.0255$, $a^*_3=0.0041$, $u^* =2.3306$. }
\label{fig:SFig_4}
\end{figure}

\vspace*{-1.0cm}

\clearpage

\begin{figure}[h]
\centering
\includegraphics[width=0.95\textwidth]{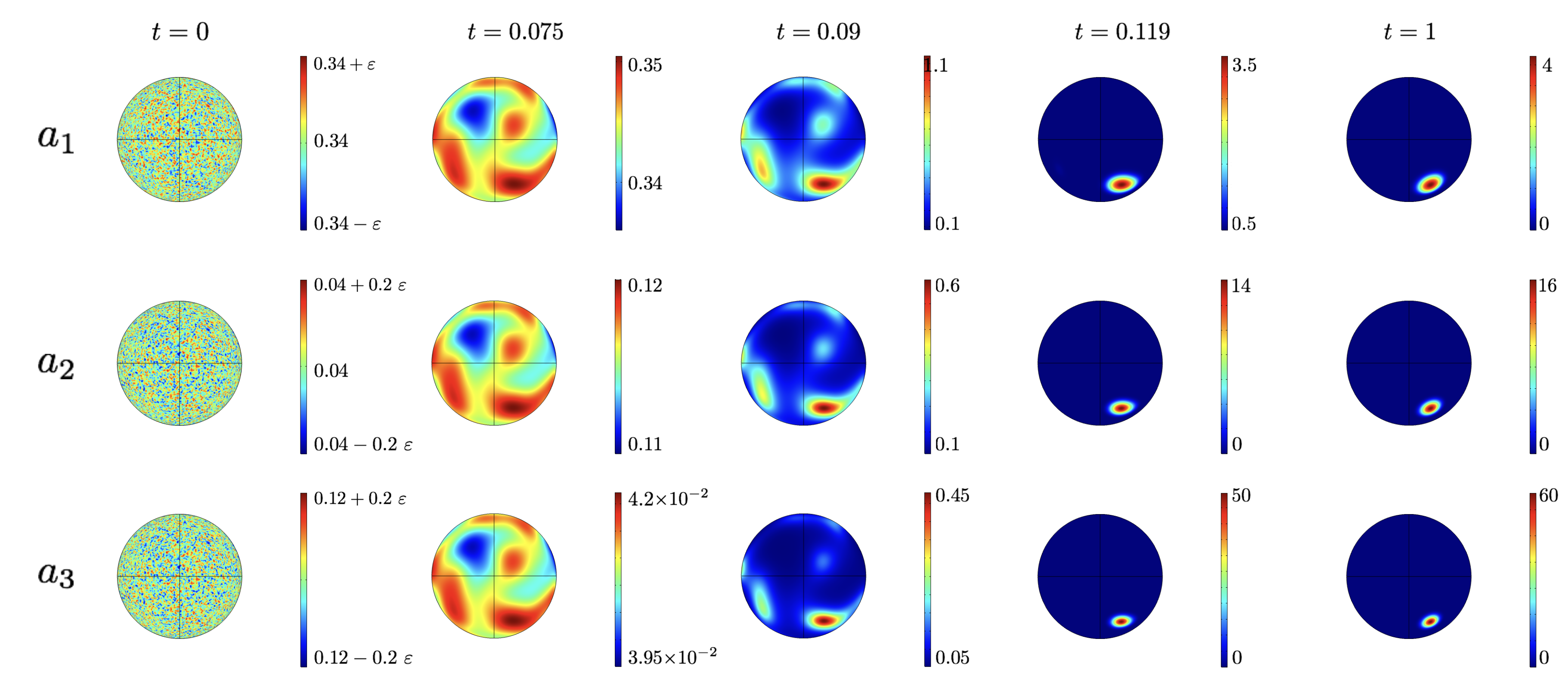}
\caption{\footnotesize\textbf{Temporal Evolution and pattern formation ($N=3$).} Spatial distribution of the three components ($a_1$, $a_2$, and $a_3$) at different non-dimensional times.  At $t=0$, a random perturbation of magnitude $\varepsilon = 10^{-10}$ is applied to the unstable homogeneous steady-state. At  $t= 0.075$, a small gradient emerges. At $t=0.09$, multiple patches can be seen  and at $t=0.119$ the system exhibits the single-patch profile. Finally, at $t= 1$, we show the single-patch steady-state. In the figure, we fixed $d_2=d_3=1$, $\gamma=1000$, $k_m=k_g=k_2=k_3=1$, $k_0=0.06$, $k_b=8$ (Region 1). The steady-state is given by $a_1(0) =0.3439$, $a_2(0) = 0.1183$, $a_3(0)=0.0407$, and $u(0) = 0.8922$.}
\label{fig:SFig_5}
\end{figure}

\vspace*{0.5cm}

\begin{figure}[h]
\centering
\includegraphics[width=0.6\textwidth]{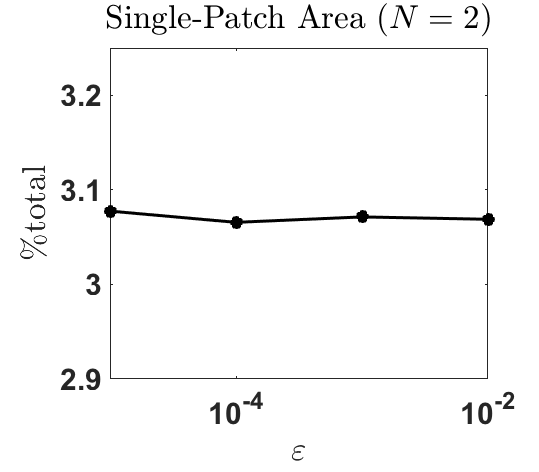}
\caption{\footnotesize\textbf{ Single-patch area and perturbation magnitude} We plot the percentage of $\mathcal{S}^\varepsilon_j$ with respect to the total surface area. We observe that such quantity does not change significantly as $\varepsilon$ changes. In this figure, we assume $R=1$, $U=13$, $A=13$, $k_0=0.0161$, $k_b=k_m=1$, $k_2=0.4409$, $a_1(0)=0.918$,         $a_2(0)=0.0191$, $\gamma=1000$, $N=2$, and  $u(0)=2.6099$}
\label{fig:SFig_6}
\end{figure}

\vspace*{-1.0cm}

\end{document}